\documentclass[sigconf,screen]{acmart}
\usepackage{xcolor}
\usepackage{xspace}
\usepackage{amsthm}
\usepackage{amsmath}
\usepackage{balance}
\usepackage{graphicx}
\usepackage{booktabs}
\usepackage{multirow}
\usepackage{enumitem}
\usepackage[caption=false,labelformat=simple]{subfig}
\usepackage[linesnumbered,ruled,vlined,norelsize]{algorithm2e}

\usepackage[utf8]{inputenc} 
\usepackage[T1]{fontenc}    
\usepackage{microtype}      

\graphicspath{{figs/}}

\hypersetup{pdfstartview=FitH}

\newcommand{\abs}[1]{\lvert#1\rvert}
\newcommand{\spara}[1]{\vspace{2mm}\noindent\textbf{#1.}}

\newcommand{\RNum}[1]{\uppercase\expandafter{\romannumeral #1\relax}}
\newcommand{\Rnum}[1]{\expandafter{\romannumeral #1\relax}}

\newcommand{\ie}{{i.e.,}\xspace}
\newcommand{\eg}{{e.g.,}\xspace}

\newcommand{\E}{\mathbb{E}}

\newcommand{\diff}{\ensuremath{\mathop{}\!\mathrm{d}}}

\newcommand{\e}{{\ensuremath{\mathrm{e}}}}

\DeclareMathOperator{\Beta}{Beta}

\newcommand{\eat}[1]{\ignorespaces}
\newcommand{\compilehidecomments}{false}
\ifthenelse{ \equal{\compilehidecomments}{true} }{%
	\newcommand{\jing}[1]{}
}{
	\newcommand{\jing}[1]{{\color{red}  [\text{Jing:} #1]}}
}

\ifthenelse{ \equal{\compilehidecomments}{true} }{%
	\newcommand{\hym}[1]{}
}{
	\newcommand{\hym}[1]{#1}
}

\ifthenelse{ \equal{\compilehidecomments}{true} }{%
	\newcommand{\revise}[1]{}
}{
	\newcommand{\revise}[1]{{\color{red} #1}}
}

\ifthenelse{ \equal{\compilehidecomments}{true} }{%
	\newcommand{\delete}[1]{}
}{
	
}

\ifthenelse{ \equal{\compilehidecomments}{true} }{%
	\newcommand{\delete}[1]{}
}{
	\newcommand{\delete}[1]{{\color{red}  [\text{Delete:} #1] #1}}
}
\setcopyright{rightsretained}

\copyrightyear{2021}
\acmYear{2021}
\acmConference[SIGMOD '21]{Proceedings of the 2021 International Conference on Management of Data}{June 20--25, 2021}{Virtual Event, China}
\acmBooktitle{Proceedings of the 2021 International Conference on Management of Data (SIGMOD '21), June 20--25, 2021, Virtual Event, China}
\acmPrice{15.00}
\acmDOI{10.1145/3448016.3457285}
\acmISBN{978-1-4503-8343-1/21/06}
\begin{document}
\title[Do the Rich Get Richer? Fairness Analysis for Blockchain Incentives]{Do the Rich Get Richer?\\ Fairness Analysis for Blockchain Incentives}
\thanks{A short version of the paper will appear in 2021 International Conference on Management of Data (SIGMOD '21), June 20--25, 2021, Virtual Event, China. ACM, New York, NY, USA, 14 pages. \url{https://doi.org/10.1145/3448016.3457285}}

\author{Yuming Huang}
\authornotemark[1]
\affiliation{%
	\institution{National University of Singapore}
	\country{}
	\city{}
}
\email{huangyuming@u.nus.edu}

\author{Jing Tang}
\orcid{0000-0002-0785-707X}
\authornote{Yuming Huang and Jing Tang contributed equally. Corresponding author: Jing Tang.}
\affiliation{%
	\institution{National University of Singapore}
	\country{}
	\city{}
}
\email{isejtang@nus.edu.sg}

\author{Qianhao Cong}
\affiliation{%
	\institution{National University of Singapore}
	\country{}
	\city{}
}
\email{cong_qianhao@u.nus.edu}

\author{Andrew Lim}
\affiliation{%
	\institution{National University of Singapore}
	\country{}
	\city{}
}
\email{isealim@nus.edu.sg}

\author{Jianliang Xu}
\affiliation{%
	\institution{Hong Kong Baptist University}
	\country{}
	\city{}
}
\email{xujl@comp.hkbu.edu.hk}

\begin{abstract}
Proof-of-Work (PoW) is the most widely adopted incentive model in current blockchain systems, which unfortunately is energy inefficient. Proof-of-Stake (PoS) is then proposed to tackle the energy issue. The rich-get-richer concern of PoS has been heavily debated in the blockchain community. The debate is centered around the argument that whether rich miners possessing more stakes will obtain higher staking rewards and further increase their potential income in the future.
In this paper, we define two types of fairness, i.e., \emph{expectational} fairness and \emph{robust} fairness, that are useful for answering this question. In particular, expectational fairness illustrates that the expected income of a miner is proportional to her initial investment, indicating that the expected return on investment is a constant. To better capture the uncertainty of mining outcomes, robust fairness is proposed to characterize whether the return on investment concentrates to a constant with high probability as time evolves.
Our analysis shows that the classical PoW mechanism can always preserve both types of fairness as long as the mining game runs for a sufficiently long time. Furthermore, we observe that current PoS blockchains implement various incentive models and discuss three representatives, namely ML-PoS, SL-PoS and C-PoS. We find that (i) ML-PoS (e.g., Qtum and Blackcoin) preserves expectational fairness but may not achieve robust fairness, (ii) SL-PoS (e.g., NXT) does not protect any type of fairness, and (iii) C-PoS (e.g., Ethereum 2.0) outperforms ML-PoS in terms of robust fairness while still maintaining expectational fairness. Finally, massive experiments on real blockchain systems and extensive numerical simulations are performed to validate our analysis.  
\end{abstract}


%

\keywords{Blockchain; Incentive; Fairness; PoW; PoS}

\maketitle

\begin{sloppy}
\section{Introduction}\label{sec:intro}
\subsection{Background}
Since 2008, blockchain has attracted a plethora of interests from both academia and industry. Essentially, a blockchain is a decentralized public ledger that contains all current and historical updates in the system. Since blockchain can achieve community trust without third parties, it can be used in numerous real-world applications, such as cryptocurrency, smart contract, voting and bidding systems. 

Particularly, Bitcoin \cite{nakamoto2008bitcoin} is the first and the most popular permission-less blockchain system, which stores user transactions within sequentially linked blocks. New transactions are confirmed when they are packed into newly generated blocks. The block generation process is called mining and a network node that potentially produces blocks is known as a miner. To generate new blocks, miners should find a solution to a cryptographic puzzle, referred to as Proof-of-Work (PoW)~\cite{jakobsson1999proofs}. To incentivize the maintenance of the Bitcoin network, the miners who successfully contribute a valid proof will be rewarded. As a consequence, miners compete with one another by operating more computational devices to seek more incentives. Such a competition, unfortunately, incurs striking energy consumption. Presently, the total electricity consumption by Bitcoin miners amounts to $74$ TWh, exceeding Austria ($70$ TWh), Switzerland ($58$ TWh) and Singapore ($47$ TWh)~\cite{BTCelectic,eleclist}. For sustainable development, Proof-of-Stake (PoS)~\cite{PPcoin} is proposed as an alternative of PoW. Instead of competing for computation power, PoS miners are more likely to win if they possess more stakes (\eg~cryptocurrency). Therefore, the PoS protocol eliminates electricity consumption and significantly reduces waste of electricity.

However, the blockchain community raises concerns about the PoS protocol that it may make the rich become richer because rich miners possessing more stakes are likely to obtain more rewards and further increase their potential income in the future~\cite{richricherbloom,ecopos}. Such a rich-get-richer phenomenon (also known as the Matthew effect) is obviously unfair to miners and eventually harms decentralization of the network when the majority of stakes are controlled by a few rich miners \cite{rosenfeld2014analysis,kiffer2018better,bonneau2015sok,Kwon:2019}. In particular, resource accumulation may increase the risk of transactions rollback and data tampering, damaging to data reliability and integrity, \eg~double spending. Despite the young age of blockchain, resource accumulation has caused several severe accidents. For example, very recently in August 2020, the transactions in Ethereum Classic were rollbacked because of a $51\%$ attack, resulting in a loss of $1.68$ million dollars \cite{ETC_doublespend}. To our knowledge, the rich-get-richer concern of the PoS protocol has been scarcely studied though it has been heatedly debated. We aim to tackle this issue by leveraging the notion of \emph{fairness}, which is one of the principles of fairness, accountability and transparency (FAT) of responsible data science \cite{Getoor_RDS_2019}.

\vspace{-2mm}
\subsection{Contributions}
In this study, we focus on the fairness of incentive models for PoW and PoS protocols, which is the key to address the rich-get-richer concern of PoS. We formally define two types of fairness, \ie~\emph{expectational} fairness and \emph{robust} fairness. Specifically, expectational fairness characterizes the relation between the resource controlled by a miner and her expected reward. However, this naive definition might be insufficient to capture the uncertainty of reward in the real world. As an example, suppose that a miner initially possesses $20\%$ of the entire stakes and consider two possible mining games. That is, the miner always receives $20\%$ of the total rewards in the first game, whereas she wins all rewards with a probability of $20\%$ and gets nothing with the remaining  probability in the second game. Apparently, the expected rewards obtained from both games are exactly the same, \ie~both games are fair in expectation. However, the reward allocation in the second game is obviously more uncertain, thereby increasing the risk of her income. To reveal the uncertainty of reward, we further propose a novel concept of robust fairness that describes the relation between the actual rewards obtained by a miner from a random outcome with respect to the stochastic process of a mining game and her initial investment.


Based on our definitions of fairness, we conduct analysis on PoW and PoS incentives. In PoW, the selection of a block proposer is based on the hash power controlled by each miner, which is independent of previous mining outcomes. In particular, a miner proposes a block with a probability being  proportional to the hash power controlled by her, and we find that this mechanism can ensure both types of fairness. Our analysis on the PoW incentive model is confirmed with experiments on Geth client (v1.9.11) \cite{Gethcode}.

In PoS, stakes serve as a competing resource and the probability of a miner proposing a new block is based on her current stakes which depend not only on her initial investment but also on the rewards received in previous mining outcomes. Moreover, there are various implementations of PoS incentives resulting in different types of behavior. In this paper, we analyze three representative PoS incentive models, namely multi-lottery PoS  (ML-PoS) for Qtum \cite{Qtum} and Blackcoin \cite{Blackcoin}, single-lottery PoS (SL-PoS) for NXT \cite{NXT} and compound PoS (C-PoS) for Ethereum 2.0 \cite{ETH20},  which cover the present popular blockchain systems.

Specifically, ML-PoS enables the probability of a miner proposing a new block being proportional to her current possessed stakes, which preserves expectational fairness. However, the return on investment of a miner under this protocol may not concentrate to a constant even after the mining game runs for a long time because of the accumulated effect of the Markov chain. Our experimental evaluations on real blockchain systems and numerical simulations reveal that allocating more initial stakes in the early stage of mining process and/or reducing the reward of each block are helpful to improve robust fairness.
Meanwhile, SL-PoS is devised by leveraging a single-lottery scheme to determine a block proposer, unlike ML-PoS that uses a multi-lottery scheme. However, we find that SL-PoS can accomplish neither expectational fairness nor robust fairness as rich miners have a higher return on investment. In fact, with the advantages accumulated during mining, the game will run to monopolization almost surely, incurring the Matthew effect. 
Recently, C-PoS introduces an additional inflation reward by distributing base incentives to every miner proportional to their possessed stakes. Our analysis shows that in addition to preserving expectational fairness, such an inflation reward is useful for reducing the uncertainty of mining income so that robust fairness is more likely to be achieved.

In summary, our main contributions are as follows.
\begin{enumerate}[topsep=1mm, partopsep=0pt, itemsep=1mm, leftmargin=18pt]
    \item We propose two types of fairness for blockchain incentives, \ie~expectational fairness and robust fairness, to characterize the relation between the resource controlled by a miner and her rewards obtained from a mining game. 
    \item We conduct a thorough theoretical analysis on fairness for the most widely adopted PoW incentive protocol and three popular PoS incentive protocols. In the context of fairness, these protocols are generally ranked in the following descending order: PoW, C-PoS, ML-PoS and  SL-PoS.
    \item We carry out extensive experiments on real blockchain systems and numerical simulations to evaluate the fairness of different incentive protocols. Experimental results confirm our theoretical findings and shed light on the future development of fair incentive protocols.
\end{enumerate}
\vspace{-2mm}
\subsection{Organizations}
Section~\ref{sec:preliminaries} discusses the details of the PoW incentive protocol and three PoS incentive protocols. Section~\ref{sec:expectationfariness} analyzes expectational fairness and Section~\ref{sec:robustfair} studies robust fairness for the mentioned four incentive protocols. Section~\ref{sec:eval} carries out experimental evaluations and numerical simulations. Section~\ref{sec:discussion} discusses the lessons we learned from fairness analysis. Section \ref{sec:relate} reviews related work, and Section \ref{sec:con} concludes the paper.

\vspace{-2mm}
\section{Incentive Models of Blockchains}\label{sec:preliminaries}
Miners compete for proposing a valid block to append it to the current blockchain, which is incentivized by rewards. The chance of winning the competition usually depends on the resource controlled by miners, \eg~computation power and staking power. In this section, we introduce the incentive protocols of Proof-of-Work (PoW) and Proof-of-Stake (PoS) that are widely adopted by many popular blockchain systems. 

\subsection{PoW Incentive Model}\label{subsec:PoW}
A new block is accepted by a PoW network after a cryptographic puzzle is solved by miners~\cite{nakamoto2008bitcoin}. Specifically, a PoW puzzle is to find a valid $\mathtt{nonce}$ such that $\operatorname{Hash}(\mathtt{nonce},\dotsc) < D $,\footnote{Apart from the argument $\mathtt{nonce}$, the $\operatorname{Hash}(\cdot)$ function has some other arguments like the merkle root of packed transactions and the hash value of the previous block, \eg~$\operatorname{Hash}(\mathtt{nonce},\mathtt{merkle~root},\mathtt{previous~hash})$.} where $D$ is a pre-defined mining difficulty. The $\operatorname{Hash}(\cdot)$ function maps a $\mathtt{nonce}$ to an integer uniformly distributed in the range $[0,2^{256}-1]$. As a result, the event that the hash value is less than $D$ meets Bernoulli distribution with a success probability of $\frac{{D}}{2^{256}}$. While $D$ is much less than $2^{256}$ by design, it is almost impossible to solve the puzzle in one single trial. The likelihood of a miner solving the puzzle depends on the number of $\mathtt{nonce}$ per unit time she can check. For instance, we assume that there are two miners $A$ and $B$ who can verify $H_A$ and $H_B$ hashes every day, respectively. Therefore, the number of valid blocks found by $A$ (resp.\ $B$) during one day is very well approximated to a Poisson random variable with mean $\mu_{A}=\frac{{DH_A}}{2^{256}}$ (resp.\ $\mu_{B}=\frac{{DH_B}}{2^{256}}$). Then, the probability that miner $A$ (resp.\ $B$) will win the next block can be computed as follows. Specifically, let $T_A$ (resp.\ $T_B$) be the arrival time of the next block for $A$ (resp.\ $B$). With respect to the Poisson process, it is well known that $T_A$ and $T_B$ follow negative exponential distributions with rate parameters $\mu_A$ and $\mu_B$~\cite[Section 2.2]{ross1996stochastic}, respectively, \ie~the probability density function of $T_A$ is $f(t_A; \mu_A)=\mu_A\e^{-\mu_At_A}$ for $t_A\geq 0$. If miner $A$ wins the next block, $A$ must find a valid block earlier than $B$ such that $T_A < T_B$. 
Therefore, the probability of miner $A$ winning the next block is 
\begin{align*}
    \Pr[T_A < T_B]
    &=\int_{0}^{\infty}\int_{0}^{t_B}\mu_A\e^{-\mu_At_A}\mu_B\e^{-\mu_Bt_B}\diff t_A\diff t_B\\
    &=\tfrac{\mu_A}{\mu_A+\mu_B}=\tfrac{H_A}{H_A+H_B}.
\end{align*}

\subsection{Multi-Lottery PoS Incentive Model}\label{subsec:ML-PoS}
Despite of its popularity, PoW incurs massive energy consumption as the mining competition relies on computation power. To eradicate the waste of computation resource, PoS protocols are invented, where the competition depends on staking power instead. In the following, we introduce the multi-lottery PoS incentive model~\cite{Blackcoin,Qtum}, referred to as ML-PoS. Stakeholders of ML-PoS blockchains create a valid block if a candidate block satisfies the condition that $\operatorname{Hash}(\mathtt{time},\dotsc)<D\cdot\mathtt{stake}$, where $\mathtt{time}$ represents the timestamp when the candidate block is generated, $D$ is a pre-determined mining difficulty and $\mathtt{stake}$ is the value of stakes possessed. Since $\operatorname{Hash}(\cdot)$ is uniformly distributed in the range $[0,2^{256}-1]$, the event that a candidate block becomes valid meets Bernoulli distribution with a success probability of $\frac{D\cdot\mathtt{stake}}{2^{256}}$. Thus, if a miner possesses more stakes, she is more likely to create a new block successfully. Moreover, miners will try at different timestamps until a candidate block becomes valid. The trials are independent of timestamps following the same Bernoulli distribution. Again, we assume that there are two miners $A$ and $B$ possessing $S_A$ and $S_B$ stakes, respectively. We refer to $T_A$ (resp.\ $T_B$) as the number of timestamps miner $A$ (resp.\ $B$) has checked until $A$ (resp.\ $B$) meets the first success timestamp. It is easy to see that $T_A$ and $T_B$ follow geometric distributions with probability parameters $p_A=\frac{D S_A}{2^{256}}$ and $p_B=\frac{D S_B}{2^{256}}$, respectively, \ie~$\Pr[T_A=t]=(1-p_A)^{t-1}p_A$. Furthermore, if $A$ wins the next block, $A$ finds a valid block with fewer timestamps than $B$ such that $T_A < T_B$ or $A$ has a chance of $50\%$ when $T_A=T_B$ to break the tie. It is easy to get that
\begin{align*}
    \Pr[T_A<T_B]
    &=\sum\nolimits_{t_B=1}^{\infty}\sum\nolimits_{t_A=1}^{t_B-1}(1-p_A)^{t_A-1}p_A(1-p_B)^{t_B-1}p_B\\
    &=\frac{p_A - p_Ap_B}{p_A+p_B-p_Ap_B}.
\end{align*}
Similarly, $\Pr[T_A=T_B]=\frac{p_Ap_B}{p_A+p_B-p_Ap_B}$. Therefore, the probability of $A$ winning the next block is 
\begin{equation*}
    \Pr[T_A<T_B]+\frac{1}{2}\cdot\Pr[T_A=T_B]=\frac{p_A - p_Ap_B/2}{p_A+p_B-p_Ap_B}.
\end{equation*}
Moreover, the time interval between two blocks is around 5--10 minutes by design. Thus, $p_A$ and $p_B$ are sufficiently small (\eg~$1/1200$), which indicates that $p_Ap_B$ is negligible. As a result, $A$ wins the next block with a probability of $\frac{p_A}{p_A+p_B}=\frac{S_A}{S_A+S_B}$.

\spara{Remark}\label{rmk:ML-PoS}
In ML-PoS, the $\operatorname{Hash}(\cdot)$ function depends on $\mathtt{time}$ (\ie~timestamp) instead of $\mathtt{nonce}$ as in PoW. Using $\mathtt{time}$ ensures that each miner has exactly one trial at each timestamp. Hence, the number of trials depends only on staking power so that the mining completion is independent of computation power. However, if $\mathtt{nonce}$ is applied, miners might try different $\mathtt{nonces}$ at each timestamp. As a result, the number of trials would rely on computation power as well.

\subsection{Single-Lottery PoS Incentive Model}\label{subsec:SL-PoS}
Another variant of PoS incentive model uses the single-lottery protocol~\cite{NXT}, referred to as SL-PoS. Unlike ML-PoS where multiple trials are involved at different timestamps when miners compete for a block, SL-PoS only allows a single trial for each block. Specifically, each miner is assigned a lottery ticket represented by $\mathtt{time}$, which is given by $\mathtt{time} = \mathtt{basetime} \cdot \operatorname{Hash}(\mathtt{pk},\dotsc)/\mathtt{stake}$, where $\mathtt{basetime}$ is a pre-determined constant, $\mathtt{pk}$ denotes the miner's public key, and $\mathtt{stake}$ refers to the miner's staking power. The protocol works as follows: (i) $\mathtt{time}$ determines when the candidate block will become valid, and (ii) the first valid block (\ie~the one with the smallest value of $\mathtt{time}$) will be accepted by the current blockchain whereas the other candidates will be discarded. Again, the miner who possesses more stakes has a better chance to get a smaller value of $\mathtt{time}$ and hence she is more likely to be selected as the proposer of the next block. Consider the two-miner scenario where miner $A$ and miner $B$ control $S_A$ and $S_B$ stakes, respectively. Without loss of generality, we assume that $S_A\leq S_B$. Let $T_A$ (resp.\ $T_B$) denote the waiting time of $A$'s (resp.\ $B$'s) candidate block becoming valid, \ie~$T_A=\mathtt{basetime} \cdot\operatorname{Hash}(\mathtt{pk}_A,\dotsc)/S_A$ where $\mathtt{pk}_A$ is $A$'s public key and $\operatorname{Hash}(\mathtt{pk}_A,\dotsc)$ is a random number uniformly distributed in the range  $[0,2^{256}-1]$ with respect to $\mathtt{pk}_A$. If $A$ wins the next block, the waiting time of $A$ should be smaller than that of $B$ such that $T_A<T_B$ or $A$ has a chance of $50\%$ when $T_A=T_B$. Therefore, the probability of miner $A$ winning the next block is
\begin{align}
    &\Pr[{T_A}< {T_B}]+\frac{1}{2}\cdot\Pr[{T_A}={T_B}] \nonumber\\
    &=\sum\nolimits_{h_B=0}^{2^{256}-1}\sum\nolimits_{h_A=0}^{\frac{S_Ah_B}{S_B}-1} \frac{1}{2^{256\cdot 2}}+\frac{1}{2}\sum\nolimits_{h_B=0}^{2^{256}-1} \frac{1}{2^{256\cdot 2}} \nonumber\\
    &=\frac{S_A}{2S_B}\cdot \frac{2^{256}-1}{2^{256}} + \frac{1}{2\cdot 2^{256}} \approx \frac{S_A}{2S_B},\label{prob:SL-PoS}
\end{align}
where $h_A=\operatorname{Hash}(\mathtt{pk}_A,\dotsc)$ and $h_B=\operatorname{Hash}(\mathtt{pk}_B,\dotsc)$.

\spara{Discussion} In PoW, ML-PoS and SL-PoS, the reward is determined by the likelihood of a miner proposing a new block. However, unlike PoW and ML-PoS, we find that the success probability of a miner proposing a new block in SL-PoS is not proportional to her staking power in general. In particular, the above analysis shows that the probability of miner $A$ winning the next block is $\frac{S_A}{2S_B}<\frac{S_A}{S_A+S_B}$ when $S_A< S_B$, \eg~$\frac{S_A}{2S_B}\approx \frac{1}{2}\cdot\frac{S_A}{S_A+S_B}$ when $S_A\ll S_B$.

\subsection{Compound PoS Incentive Model}\label{subsec:C-PoS}
Recently, a compound PoS incentive protocol, referred to as C-PoS, is deployed by the next generation Ethereum 2.0~\cite{ETH20}. Miners will be rewarded as (i) proposers who propose a new block and (ii) attesters who verify the validity of a block. Specifically, a miner will be assigned one identity for every 32 Ethers deposited in the smart contract. These identities are randomly and disjointly partitioned to 32 shards as attesters to verify transactions in parallel. In addition, one identity will be selected uniformly at random from every shard as the block proposer. During each mining epoch, a miner can receive $3 \cdot \mathtt{base} \cdot \mathtt{vote}$ incentives for each attester identity she controls, where $\mathtt{base}$ is a pre-determined constant, $\mathtt{vote}$ is the percentage of attesters that stay online and actively submit votes (which is usually close to $100\%$). Meanwhile, a total of $\frac{1}{8}\cdot \mathtt{base} \cdot N$ incentives are provided for block proposers because of their contributions to their newly proposed blocks in each epoch, where $N$ is the total number of identities assigned to all miners. Consider a two-miner scenario in a generalized C-PoS, where miners $A$ and $B$ possess $S_A$ and $S_B$ stakes, respectively. Let $v$ and $w$ denote the total rewards for attesters and proposers, respectively. Assume that there are $P$ shards. Then, miner $A$ will obtain $v\cdot \frac{S_A}{S_A+S_B}$ stakes as attesters and $w\cdot \frac{X}{P}$ stakes as proposers, where $X$ is the number of blocks proposed by $A$ in the epoch following binomial distribution $\operatorname{Bin}\big(P,\frac{S_A}{S_A+S_B}\big)$. Therefore, the total reward for $A$ is $v\cdot \frac{S_A}{S_A+S_B}+w\cdot \frac{X}{P}$, where $X\sim \operatorname{Bin}\big(P,\frac{S_A}{S_A+S_B}\big)$.

\spara{Remark} To maintain the total reward for attesters (or proposers) in each epoch stable, Ethereum 2.0 slowly decreases the value of $\mathtt{base}$ as additional identities are rewarded after each epoch. In particular, $v+w$ is around $1$ Ether and $v$ is ${\sim}20$ times of $w$ in Ethereum 2.0~\cite{ETH20}. 

\section{Expectational Fairness}\label{sec:expectationfariness}
Fairness is one of the most important concerns for the design of incentive mechanisms. Intuitively, in a blockchain system with a fair incentive mechanism, the reward of a miner should be proportional to the amount of resource (\eg~computation power in PoW and staking power in PoS) that she obtains. That is, the return on investment is identical for every miner. In this section, we first introduce our assumptions and the definition of expectational fairness and then analyze such fairness for the aforementioned four incentive protocols, namely PoW, ML-PoS, SL-PoS and C-PoS. 


\subsection{Assumption and Definition}
\label{sec:assumption}
We consider permission-less blockchains for all the protocols analyzed and make the following assumptions. 
\eat{Specifically, only the updates in the longest chain are valid, and one block proposer will be selected to win a coinbase reward for each block. Moreover, we 
also make the following assumptions:}

\begin{enumerate}[topsep=1mm, partopsep=0pt, itemsep=1mm, leftmargin=18pt]
\item Only two miners, \ie~miner $A$ and miner $B$, are competing for proposing blocks in the network.\label{as:1}
\item Initially, the resource share of miner $A$ (resp.\ $B$) is $a$ (resp.\ $b$). Without loss of generality, $a$ and $b$ are normalized such that $a+b=1$. \label{as:2}
\item The reward of each mining epoch remains the same, \ie~$w$ proposer reward and $v$ inflation reward (\eg~attester reward in C-PoS).\label{as:3} 
\item Both miners $A$ and $B$ do not perform additional action after a mining game starts.\label{as:4}
\end{enumerate}
Without loss of generality, we focus on studying the relationship between the original mining share of miner $A$ and the total rewards obtained by $A$. \hym{The simple two-miner model in Assumption~\ref{as:1} is for the sake of brevity. In Section~\ref{dis:exntend-multiple}, we elaborate how to generalize it to the scenario with multiple miners, \eg~considering $B$ as a set of miners in PoW \cite{kwon2017selfish,eyal2018majority}.} In Assumption~\ref{as:2}, the resource share represents the hash power ratio controlled by miners in PoW and the proportion of the initial amount of stakes possessed in PoS. As to Assumption~\ref{as:3}, blockchain systems may change rewards when time evolves. For example, Bitcoin halves its block reward every $210{,}000$ blocks. However, the current halving period of Bitcoin is around four years, which is a long time that can asymptotically support the assumption that the reward of each block remains unchanged \cite{garay2015bitcoin,moser2015trends}. Also note that the values of $w$ and $v$ reflect the relative relation between an initial resource and a reward per mining epoch, since $a$ and $b$ are normalized. According to Assumption~\ref{as:4}, we assume that both PoW and PoS miners passively participate in the mining game and do not perform any additional action like withdrawal or top-up \cite{kroll2013economics,ciamac2017monopoly}. 

Our aim is to study whether a miner with a fraction of the total resource can finally obtain the same fraction of reward in expectation. To achieve this goal, we leverage the concept of expectational fairness which is formally defined as follows.

\begin{definition}[Expectational Fairness] An incentive mechanism preserves expectational fairness for miner $A$ possessing a fraction $a$ of the total resource if $A$ receives a fraction $\lambda_A$ of the total reward satisfying $\mathbb{E}[\lambda_A] = a$.
\end{definition}


\subsection{Expectational Fairness for PoW}
Initially, miner $A$ (resp.\ miner $B$) controls a fraction $a$ (resp.\ $b=1-a$) of the total hash power (\ie mining rigs). The miner who successfully creates a new block will be rewarded an incentive $w$. Let $\lambda_A$ represent the fraction of rewards received by miner $A$ after a total of $n$ blocks are appended to the blockchain. The total number $n\lambda_A$ of blocks proposed by $A$ follows a binomial distribution $\operatorname{Bin}(n,a)$. Therefore, the expected reward of miner $A$ is $nwa$ and hence the expectation of the reward fraction is always equal to $a$. This indicates that PoW achieves expectational fairness, \ie~the expected reward of miner $A$ is proportional to her initial computation power.

\begin{theorem}\label{thm:PoW-fair}
PoW achieves expectational fairness.
\end{theorem}

\subsection{Expectational Fairness for ML-PoS}\label{subsec:PoS-incentive}
At the beginning, miner $A$ (resp.\ miner $B$) owns a fraction $a$ (resp.\ $b=1-a$) of the total stakes. Each block gives a reward of $w$ stakes. Unlike PoW, the chance that $A$ can win a block not only depends on the initial staking power (\ie $a$) but also relies on previous mining outcomes. Specifically, the probability of $A$ proposing a new block is determined by $A$'s current staking power, including the initial stakes and the earned stakes. For example, if a miner is ``lucky'' to mine some blocks, her expected rewards will be improved in the future as the volume of her stakes increases. In the following, we show that ML-PoS still preserves expectational fairness. 
\begin{theorem}\label{thm:ML-PoS-fair}
ML-PoS achieves expectational fairness.
\end{theorem}
\hym{Intuitively, the rationality behind Theorem~\ref{thm:ML-PoS-fair} is from the fact that ML-PoS enables the probability of a miner proposing a new block being proportional to her currently possessed stakes. That is, the conditional expected reward for miner $A$ is proportional to her currently possessed stakes for each block. Taking expectation over the randomness of her possessed stakes, we can obtain that the expected reward for miner $A$ is proportional to her expected possessed stakes for each block, which concludes the theorem. Due to space limitations, we omit all proofs, and interested readers are referred to the appendix in our technical report~\cite{Missing_Proofs} for details.}


\subsection{Expectational Fairness for SL-PoS}\label{subsec:SL-PoS-EF}
As we discussed in Section~\ref{subsec:SL-PoS}, different from ML-PoS, the probability that $A$ wins a block under SL-PoS is $\frac{a}{2b}\leq a$ when $a\leq b$. \hym{As a result, unless $a=b$, the expected reward of $A$ is not guaranteed to be proportional to her initial resource share $a$ and hence SL-PoS does not preserve expectational fairness.}
\begin{theorem}\label{thm:SL-PoS-fair}
	SL-PoS does not ensure expectational fairness even after an infinity number of blocks are proposed.
\end{theorem}

\subsection{Expectational Fairness for C-PoS}
Different from ML-PoS, in each mining epoch of C-PoS, the network randomly selects $P$ block proposers and each proposer will receive a proposer reward of $\frac{w}{P}$ (\eg~$P=32$ in Ethereum 2.0~\cite{ETH20}). In addition, the system also provides a total of $v$ inflation reward (\eg~attester reward in Ethereum 2.0~\cite{ETH20}) to all miners. Both the probability of proposer selection and the allocation of inflation reward are proportional to miners' present staking power. \hym{Therefore, analogous to ML-PoS, C-PoS still preserves expectational fairness although it is more complicated.}
\begin{theorem}\label{thm:C-PoS-fair}
C-PoS achieves expectational fairness.
\end{theorem}

\section{Robust Fairness} \label{sec:robustfair}
In the previous section, we analyze the stochastic process of the mining procedure leveraging the concept of expectational fairness. However, people usually care more about fairness in every possible outcome rather than a simple expectation. Expectational fairness, unfortunately, cannot provide such analysis. To tackle this issue, in this section, we propose a new concept of robust fairness, which can better characterize the relation between the initial investment and the reward distribution.

\subsection{Definition of Robust Fairness}
Robust fairness, intuitively, implies that the random outcome of a miner's reward is concentrated to its initial investment with a high probability. To capture the concept of robust fairness, we define $(\varepsilon,\delta)$-fairness as follows. 



\begin{definition}[$(\varepsilon,\delta)$-Fairness]\label{def:fair}
For any given pair of parameters $(\varepsilon,\delta)$ such that $\varepsilon \geq 0 $ and $0 \le \delta \le 1$, an incentive mechanism preserves an $(\varepsilon,\delta)$-fairness for miner $A$ possessing a fraction $a$ of the total resource if $A$ receives a fraction $\lambda_{A}$ of the total reward satisfying
\begin{equation*}
    \Pr\big[(1-\varepsilon)a\leq \lambda_{A}\leq (1+\varepsilon)a\big]\ge 1 - \delta.
\end{equation*}
\end{definition}

Definition \ref{def:fair} defines bicriteria fairness. Note that the definition of $(\varepsilon,\delta)$-fairness does not explicitly include the total number $n$ of blocks (or epochs for C-PoS) for competing.\footnote{To reveal $n$ explicitly, we may use $\lambda_{A,n}$ to indicate the mining outcome of $n$ blocks. For brevity, $\lambda_{A}$ refers to $\lambda_{A,n}$ in this paper unless specified otherwise.} Usually, $\lambda_{A}$ will gradually converge as long as $n$ increases. Thus, our analysis of $(\varepsilon,\delta)$-fairness is carried out on a large value of $n$, including a special case where $n$ goes to infinity. In the rest of the paper, we also say that an $\varepsilon$-fairness is achieved with a probability at least $1-\delta$, which exactly means that $(\varepsilon,\delta)$-fairness is preserved. According to Definition \ref{def:fair}, smaller values of $\varepsilon$ and/or $\delta$ indicate a higher level of fairness. In particular, an incentive mechanism preserving $(0,0)$-fairness, which is absolutely fair, is an ideal protocol. 


\eat{
In what follows, we show that PoW protocols preserve $(\epsilon,\delta)$-fairness if $n \ge \frac{\ln(\frac{2}{\delta})}{2a^2\epsilon^2}$, which indicates PoW miners satisfies robust fairness if the miner plays the mining game in sufficiently long time. In multi-lottery PoS protocols, to achieve robust fairness, it requires $\frac{1}{n} + w \le \frac{2a^2\epsilon^2}{\ln\frac{2}{\delta}}$, meaning that not only the miner should maintain mining for a long time but also the reward and transaction fee in each block should be sufficiently small. Completely different to the previous two protocols, the robust fairness of single-lottery PoS becomes worse as mining game proceeds. We show that single-lottery PoS achieves $\epsilon$-fairness with $0$ probability when $n\to \infty$.
}
\subsection{Robust Fairness for PoW}
For the PoW incentive protocol, let $F(k;n,a)$ be the cumulative distribution function of the random variable $n\lambda_A$, \ie
\begin{equation*}
    F(k;n,a):=\Pr[n\lambda_A\leq k]=\sum_{i=0}^{k}\binom{n}{i}a^{i}(1-a)^{n-i}.
\end{equation*}
In addition, let $\Delta(\varepsilon;n,a)$ be a function of $\varepsilon,n,a$ such that
\begin{equation*}
    \Delta(\varepsilon;n,a):=F(\lfloor n(1+\varepsilon)a\rfloor;n,a)-F(\lceil n(1-\varepsilon)a\rceil;n,a).
\end{equation*}
Thus, $\Delta(\varepsilon;n,a) =\Pr\big[(1-\varepsilon)a\leq \lambda_A\leq (1+\varepsilon)a\big]$. Therefore, to achieve an $(\varepsilon,\delta)$-fairness for miner $A$, the incentive parameters $n$ and $a$ must satisfy that $\Delta(\varepsilon;n,a)\geq 1-\delta$. However, $\Delta(\varepsilon;n,a)$ is complicated, which cannot clearly and explicitly reveal the requirements. \hym{To tackle this issue, we make use of Hoeffding inequality \cite{hoeffding1994probability}, which provides a \textit{neat} expression of an upper bound on the probability that the sum of bounded independent random variables deviates from its expected value by more than a certain amount. In what follows, we give a sufficient condition (but not necessary) required by PoW for preserving an $(\varepsilon,\delta)$-fairness.}
\begin{theorem}\label{thm:PoW-robust}
PoW preserves an $(\varepsilon,\delta)$-fairness for miner $A$ with the computation power of $a$ if the total number $n$ of blocks for competing satisfies $n \ge \frac{\ln(\frac{2}{\delta})}{2a^2\varepsilon^2}$. 
\end{theorem}

By Theorem \ref{thm:PoW-robust}, we know that if miner $A$ competes for more blocks and/or possesses more hash power, she will feel fairer. Moreover, it is easy to get that for PoW, $\lambda_A$ converges to $a$ almost surely when $n\to \infty$, \ie~$\Pr[\lim_{n \to \infty}\lambda_A=a]=1$. Thus, when $n\to\infty$, PoW is absolutely fair by achieving the $(0,0)$-fairness.


\subsection{Robust Fairness for ML-PoS}
The mining process of ML-PoS can be modeled by a classical P\'{o}lya Urn such that the fraction $ \lambda_A$ of blocks proposed by $A$ will finally converge to a beta distribution $\Beta(\frac{a}{w},\frac{b}{w})$ almost surely~\cite[Theorem 3.2]{mahmoud2008polya}. In particular, it is sufficient to achieve an $(\varepsilon,\delta)$-fairness if $I_{(1+\varepsilon)a}(\frac{a}{w},\frac{b}{w})-I_{(1-\varepsilon)a}(\frac{a}{w},\frac{b}{w})\geq 1-\delta$, where $I$ is the regularized incomplete beta function. However, the relation is again not explicitly revealed.  
In the following, we derive a simple sufficient requirement by ML-PoS for achieving an $(\varepsilon,\delta)$-fairness. Unlike PoW where the mining outcomes are independent and identically distributed random variables, the mining competition for ML-PoS is a Markov chain process. \hym{To tackle this issue, we leverage Azuma inequality \cite{azuma1967weighted} for martingales \cite{doob1953stochastic}, which supports certain weakly dependent random variables.}


\eat{\begin{theorem}
\label{thr:powupper}
PoW fairness is upper bounded by $\left(n,\epsilon,\exp\left(-2n\left[\frac{\epsilon a}{a+b}\right]^2\right)\right)$.
\end{theorem}}

\begin{theorem}\label{thm:ML-PoS-robust}
ML-PoS preserves an $(\varepsilon,\delta)$-fairness for miner $A$ with staking power of $a$ if the total number $n$ of blocks for competing and the block reward $w$ satisfy $1/n+w\leq \frac{2a^2\varepsilon^2}{\ln\frac{2}{\delta}}$.
\end{theorem}

If miner $A$ competes more blocks and/or possesses more staking power, it is easier for ML-PoS to achieve an $(\varepsilon,\delta)$-fairness. However, unlike PoW, ML-PoS is sensitive to the reward $w$, \eg~ a small reward for each block is more likely to be fair.

\subsection{Robust Fairness for SL-PoS}\label{subsec:SL-PoS-RF}
According to Section~\ref{subsec:SL-PoS}, for SL-PoS, the probability that miners win a block is not proportional to their staking powers. Specifically, the return on investment of a miner increases along with her staking power, which shows a clear unfairness of rich-get-richer. In the following, we study the robust fairness for SL-PoS by exploring the reward distribution.

Our analysis utilizes the techniques of Stochastic Approximation (SA)~\cite{robbins1951stochastic,renlund2010generalized}. We first introduce some useful definitions and lemmas of SA in the following. 

\begin{definition}[Stochastic Approximation~\cite{renlund2010generalized}]\label{def:SA}
A stochastic approximation algorithm $\{Z_n\}$ is a stochastic process taking value in $[0,1]$, adapted to the filtration $\mathcal{F}_n$, that satisfies, 
\begin{equation*}
    Z_{n+1} - Z_n = \gamma_{n+1}\big(f(Z_n) + U_{n+1}\big),
\end{equation*}
where $\gamma_n$, $U_n\in \mathcal{F}_n$, $f\colon [0,1] \mapsto \mathbb{R}$ and the following conditions hold almost surely
\begin{enumerate}[label=(\roman*)]
    \item $c_l/n \le \gamma_n \le c_u/n$,\label{cond1}
    \item $\abs{U_n} \le K_u$,\label{cond2}
    \item $\abs{f(Z_n)} \le K_f$, and \label{cond3}
    \item $\abs{\E[\gamma_{n+1} U_{n+1}\mid \mathcal{F}_n]} \le K_e\gamma_n^2$,\label{cond4}
\end{enumerate}
where $c_l,c_u,K_u,K_f,K_e$ are finite positive real numbers.
\end{definition}
The stochastic approximation algorithm is originally used for root-finding problems. Specifically, $\{Z_n\}$ is a stochastic process with an initial value of $Z_0$, $\gamma_n$ denotes a moving step size gradually decreasing along with $n$ and $\gamma_nU_n$ is a random noise with expectation tending to zero quickly. In a nutshell, $Z_n$ moves towards one of the zero points of $f(\cdot)$ and finally converges as long as the update process iterates a sufficiently large number of steps. 

\begin{lemma}[Zero Point of SA~\cite{renlund2010generalized}]\label{lemma:sa-zero}
If $f$ is continuous then $\lim_{n\to \infty} Z_n$ exists almost surely and is in $Q_f=\{x\colon f(x)=0\}$.
\end{lemma}

Note that $Z_n$ may not converge to every zero point in $Q_f$. That is, if a zero point $q$ is stable, $Z_n$ converges to $q$ when $n\to \infty$ with a positive probability. Otherwise, if $q$ is an unstable point, $Z_n$ converges to $q$ with zero probability. The following lemmas characterize the properties of stable and unstable points of SA.

\begin{definition}[Attainability~\cite{renlund2010generalized}]
	A subset ${I}$ is attainable if for every fixed $N \ge 0$, there exists an $n \ge N$ such that $\Pr[Z_n \in {I}] > 0$.
\end{definition}
\begin{lemma}[Stable Zero Point of SA~\cite{renlund2010generalized}]\label{lemma:sa-stable}
	Suppose $q\in Q_f$ is stable, \ie~$f(x)(x-q)<0$ whenever $x\ne q$ is close to $q$. If every neighborhood of $q$ is attainable then $\Pr[Z_n \to p] > 0$.
\end{lemma}
\begin{lemma}[Unstable Zero Point of SA~\cite{renlund2010generalized}]\label{lemma:sa-unstable}
	Assume that there exists an unstable point $q$ in $Q_f$, \ie~such that $f(x)(x-q) \ge 0$ locally, and that $\E[U_{n+1}^2\mid \mathcal{F}_n] \ge K_L$ holds, for some $K_L>0$, whenever $Z_n$ is close to $q$. Then, $\Pr[Z_n \to q] = 0$. 
\end{lemma}

Now, we are ready to analyze the robust fairness for SL-PoS. Specifically, we denote $Z_n$ as the fraction of staking power possessed by miner $A$ after $n$ blocks are competed, \eg~$Z_0=a$. We show that $\{Z_n\}$ is a stochastic approximation algorithm. In particular, the update of $Z_n$ is directed by the probability that the miner wins the next block characterized by $f(\cdot)$. Moreover, the update step size of $Z_n$ also decreases along with $n$ because the more stakes are issued during the mining game, the less one block outcome can affect $Z_n$. Then, we apply the stochastic approximation algorithm to study the asymptotic behavior of $Z_n$. Interestingly, we find that $Z_n$ will finally converge to $0$ or $1$ almost surely, which indicates that SL-PoS cannot achieve robust fairness.  
\begin{theorem}\label{thm:SL-PoS-robust}
For SL-PoS, the proportion reward $\lambda_{A}$ of miner $A$ converges to either $0$ or $1$ almost surely when $n\to\infty$. This indicates that SL-PoS cannot achieve robust fairness.
\end{theorem}

\begin{figure}[!tbp]
	\centering  
	\vspace{-2mm}
	\includegraphics[width=0.3\textwidth]{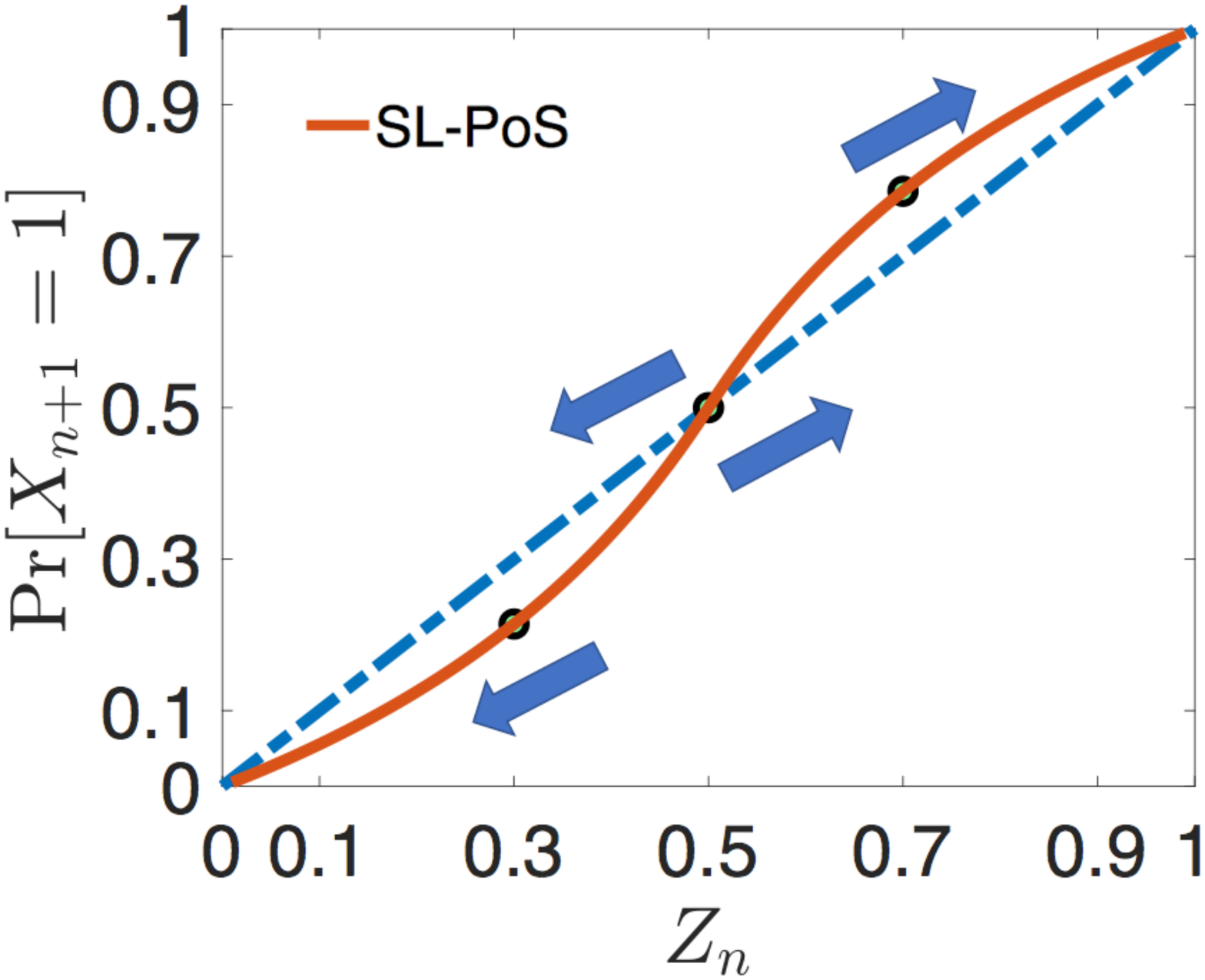}
	\vspace{-3mm}
	\caption{Probability of winning the next block for SL-PoS.} \label{fig:SL-PoS} 
	\vspace{-4mm}
\end{figure}

Theorem~\ref{thm:SL-PoS-robust} states that $\lambda_A$ will finally converge to either $\lambda_A = 0$ or  $\lambda_A = 1$, no matter how much staking power is initially controlled by miner $A$. In other words, the mining game ends with the fact that one monopoly miner acquires almost $100\%$ staking power. \figurename~\ref{fig:SL-PoS} illustrates how the fraction of staking power evolves during the mining game. When $Z_n = 0.3$, the probability miner $A$ can win another block is less than $30\%$ thus her expected fraction of stakes decreases as the mining game proceeds. Finally, $Z_n$ tends to $0$. Vice versa, when $Z_n = 0.7$, the miner will win a block with a  probability higher than $70\%$. As a consequence, her fraction of stakes tends to $1$. Specially, when $Z_n$ initiates on $0.5$, the fraction of staking power possessed by the miner gradually leaves $Z_n = 1/2$ towards either the left side or the right side and then converges to $0$ or $1$ with a fifty-fifty chance.



\subsection{Robust Fairness for C-PoS}
C-PoS provides both inflation and proposer rewards. The income uncertainty comes from the proposer reward, which is reduced by the inflation reward that is distributed proportionally to miners' shares of stakes. Therefore, compared with ML-PoS, C-PoS is more likely to achieve robust fairness. 


\begin{theorem}\label{thm:C-PoS-robust}
	C-PoS preserves an $(\varepsilon,\delta)$-fairness for miner $A$ with the staking power of $a$ if the total number $n$ of epochs for competing, the shard size $P$ of each epoch, the proposer reward $w$ and inflation reward $v$ for each mining epoch satisfy $\frac{w^2({1}/{n} + w+v)}{(w+v)^2P} \le \frac{2a^2\varepsilon^2}{\ln\frac{2}{\delta}}$.
\end{theorem}

\begin{figure*}[!hbpt]
	\centering  
	\includegraphics[width=0.6\textwidth]{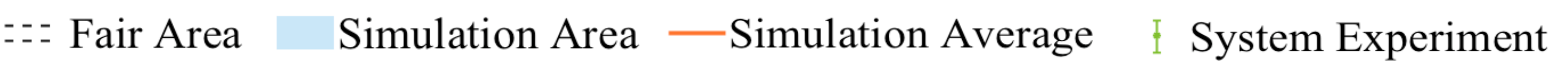}\vspace{-4mm}\\
	\subfloat[PoW]{\includegraphics[width=0.25\textwidth]{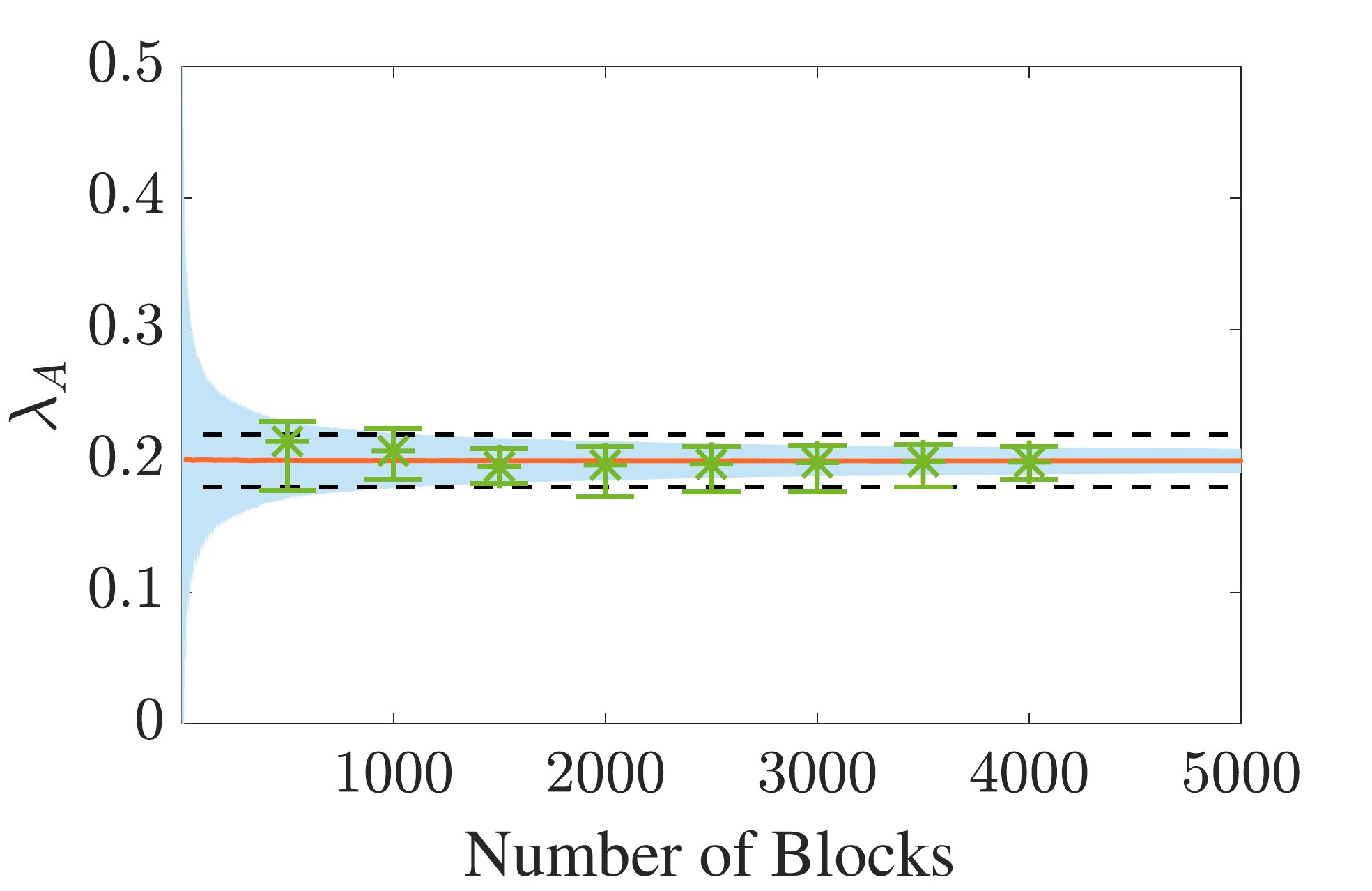}\label{fig:powsyssim}}\hfil
	\subfloat[ML-PoS ]{\includegraphics[width=0.25\textwidth]{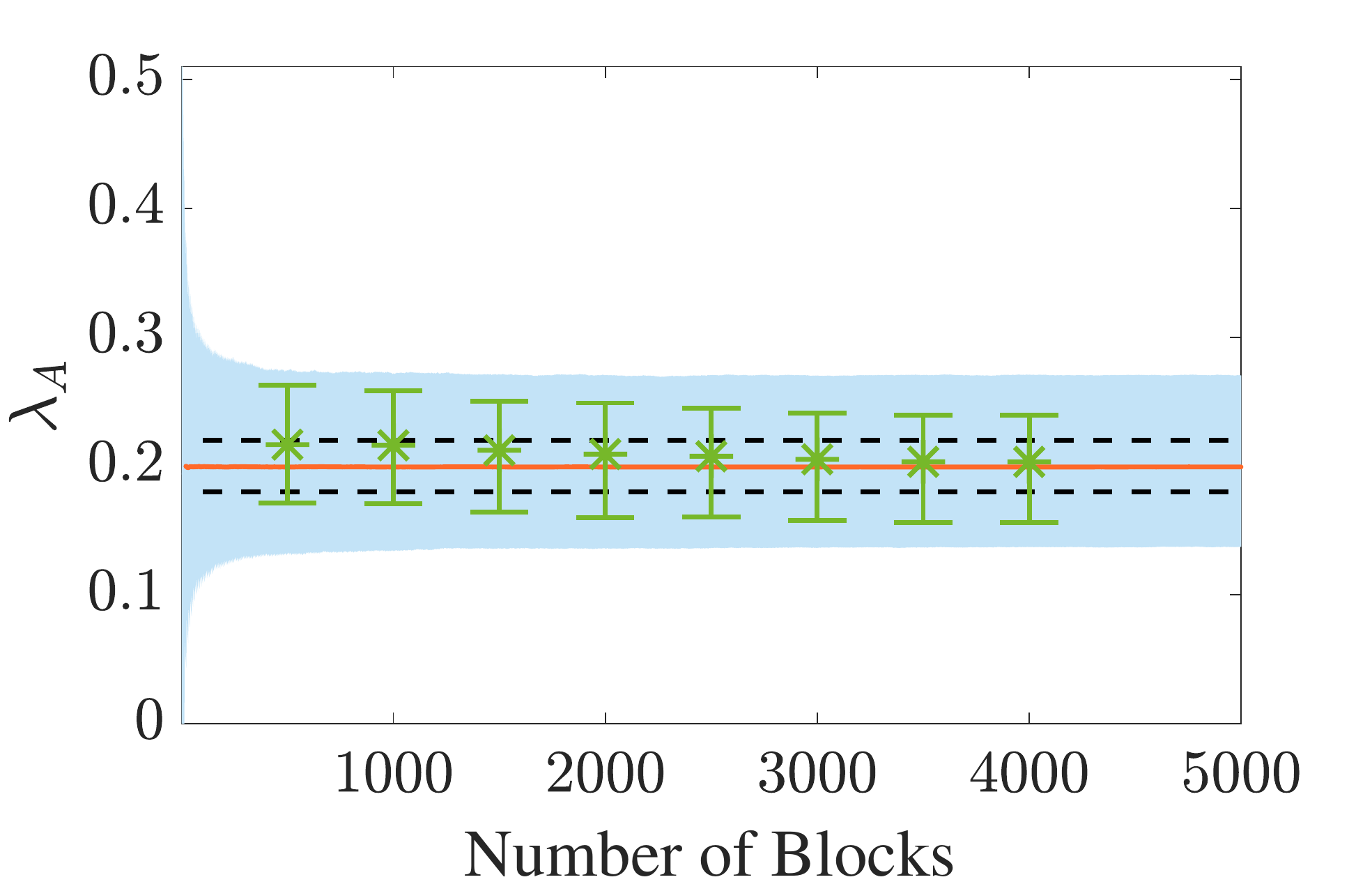}\label{fig:mlpossyssim}}\hfil
	\subfloat[SL-PoS]{\includegraphics[width=0.25\textwidth]{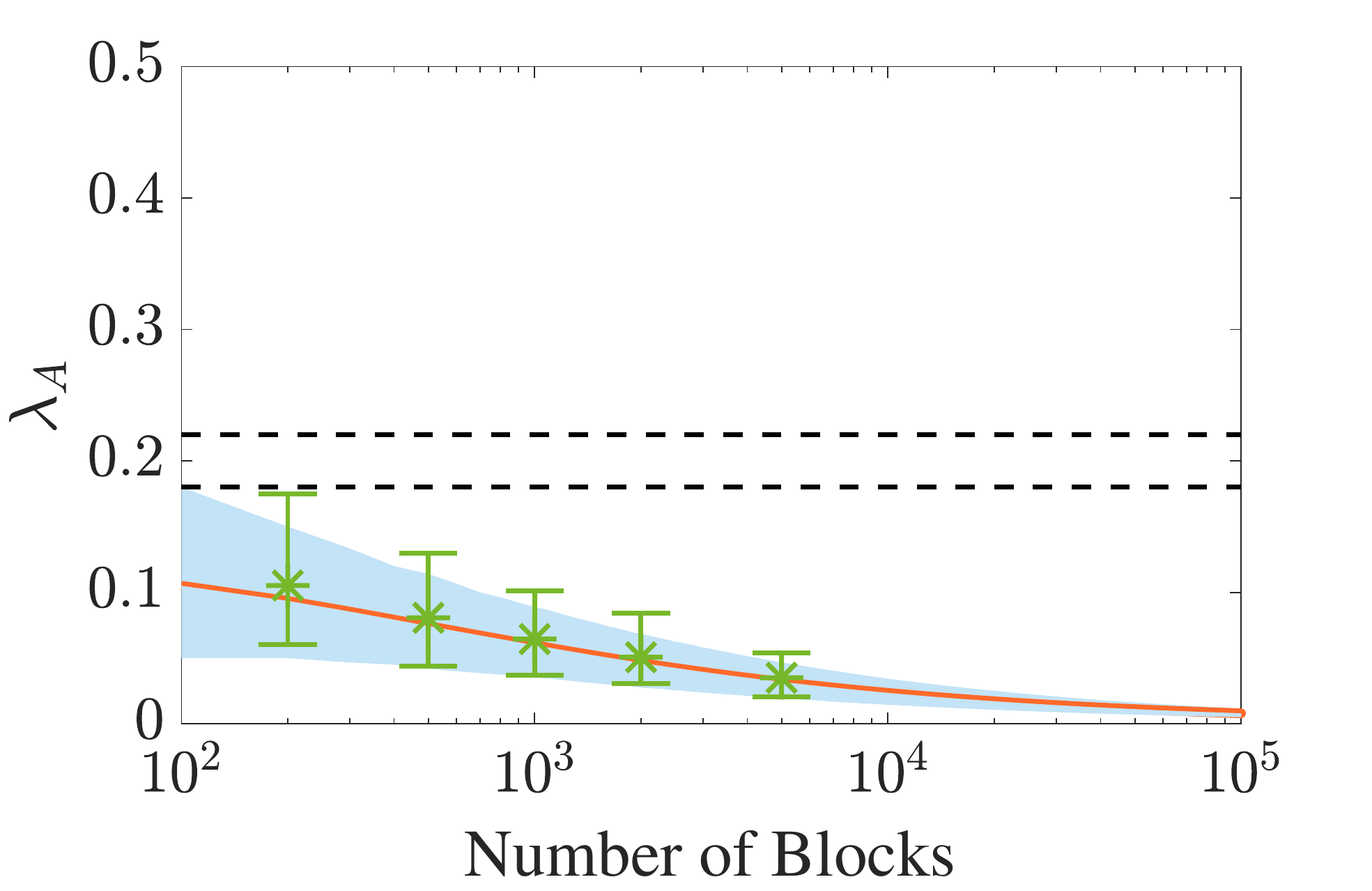}\label{fig:slpossyssim}}\hfil
	\subfloat[C-PoS]{\includegraphics[width=0.25\textwidth]{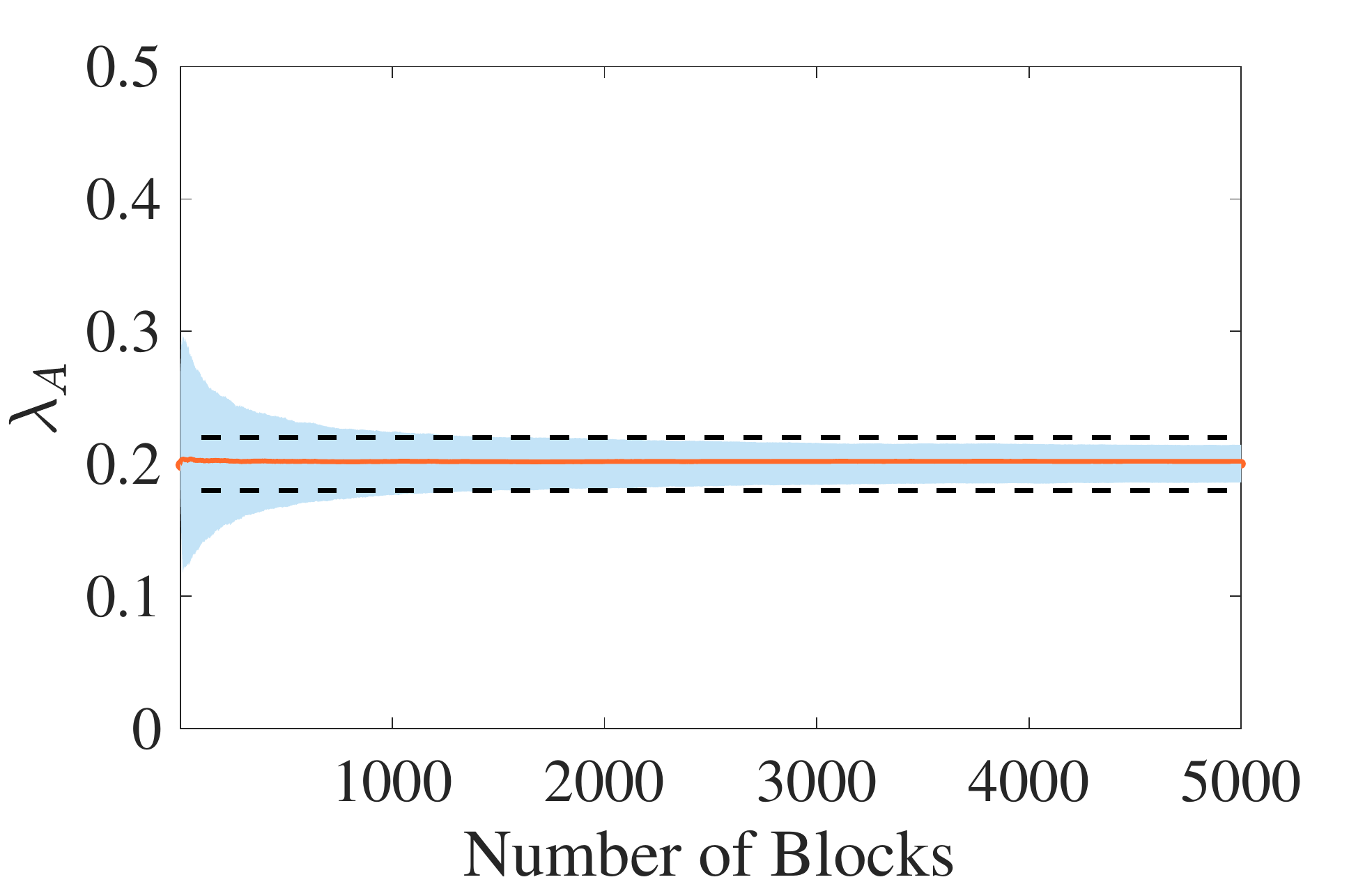}\label{fig:cpossim}}
	\vspace{-3mm}
	\caption{Evolution of $\lambda_A$ along with the number $n$ of blocks under $a = 0.2$,  $w=0.01$ and $v=0.1$.}
	\label{fig:syssim}
	\vspace{-4mm}
\end{figure*}

Theorem~\ref{thm:C-PoS-robust} states that C-PoS is more likely to achieve an $(\varepsilon,\delta)$-fairness than ML-PoS because $\frac{w^2(1/n+w+v)}{(w+v)^2P}$ is notably smaller than $1/n+w$ due to the inflation reward $v$ and $P$ shards in each epoch. In particular, to achieve robust fairness for C-PoS, we can increase inflation reward $v$ and shard size $P$, and meanwhile reduce proposer reward $w$. Note also that Theorem~\ref{thm:C-PoS-robust} degenerates to Theorem~\ref{thm:ML-PoS-robust} if no inflation reward is provided and there is one shard in each epoch, \ie~$v = 0$ and $P=1$. 
\section{Experimental Evaluation}
\label{sec:eval}
In this section, we evaluate the fairness for PoW and three PoS protocols (\ie~ML-PoS, SL-PoS and C-PoS) with both real system experiments and numerical simulations. 
\subsection{Experimental Setup}
In the real system deployment, we select Geth client (v1.9.11) \cite{Gethcode}, Qtum core (v0.19.0.1) \cite{Qtumcode} and NXT client (v1.12.2) \cite{Nxtcode} as representatives of PoW, ML-PoS and SL-PoS mechanisms, respectively. Note that C-PoS is proposed by Ethereum 2.0 \cite{ETH20} which is still under development, due to which we cannot evaluate with real system experiments. All experiments are conducted on Amazon AWS EC2 Services. Specifically, PoW experiments are deployed on M5.4xlarge instances each with a 16-core Intel Xeon Platinum 8175M CPU and 32GB RAM. Since PoS protocols are computational insensitive, ML-PoS and SL-PoS experiments are conducted on M5.large instances each with a 2-core Intel Xeon Platinum 8175M CPU and 4GB RAM. In each test case, we implement a two-miner network and each miner is deployed on an individual EC2 instance. We repeat the experiments $10$ times for PoW and $500$ times for PoS, and report the statistical results. 

Numerical simulations are also carried out to supplement the experiments especially when real system evaluations are limited in computational resource or mining time. In particular, we validate the fairness for C-PoS completely based on numerical simulations due to the lack of the real system. We repeat the simulations $10{,}000$ times and report the statistical results. 

For the evaluations of robust fairness, we set $\varepsilon = 0.1$ and $\delta = 10\%$ by default. That is, with a probability of at least $90\%$, the return on investment of a miner in a random outcome is in the range of $[0.9,1.1]$ of the average over all miners. For convenience, we refer to the range $[(1-\varepsilon)a,(1+\varepsilon)a]$ as \textit{fair area} and the range $[0,(1-\varepsilon)a)\cup((1+\varepsilon)a,1]$ as \textit{unfair area}. We measure the likelihood of $\lambda_{A}$ locating in the fair area to reveal the robust fairness, \ie~$(\varepsilon,\delta)$-fairness is achieved if such a likelihood is no less than $1-\delta$.
\vspace{-0.5mm}
\subsection{Fairness Results}
\begin{figure*}[!h]
	\centering 
	\includegraphics[width=0.9\textwidth]{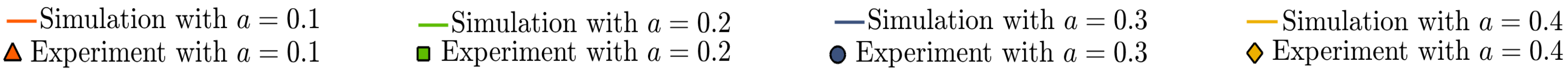}\vspace{-4mm}\\
	\subfloat[PoW]{\includegraphics[width=0.25\textwidth]{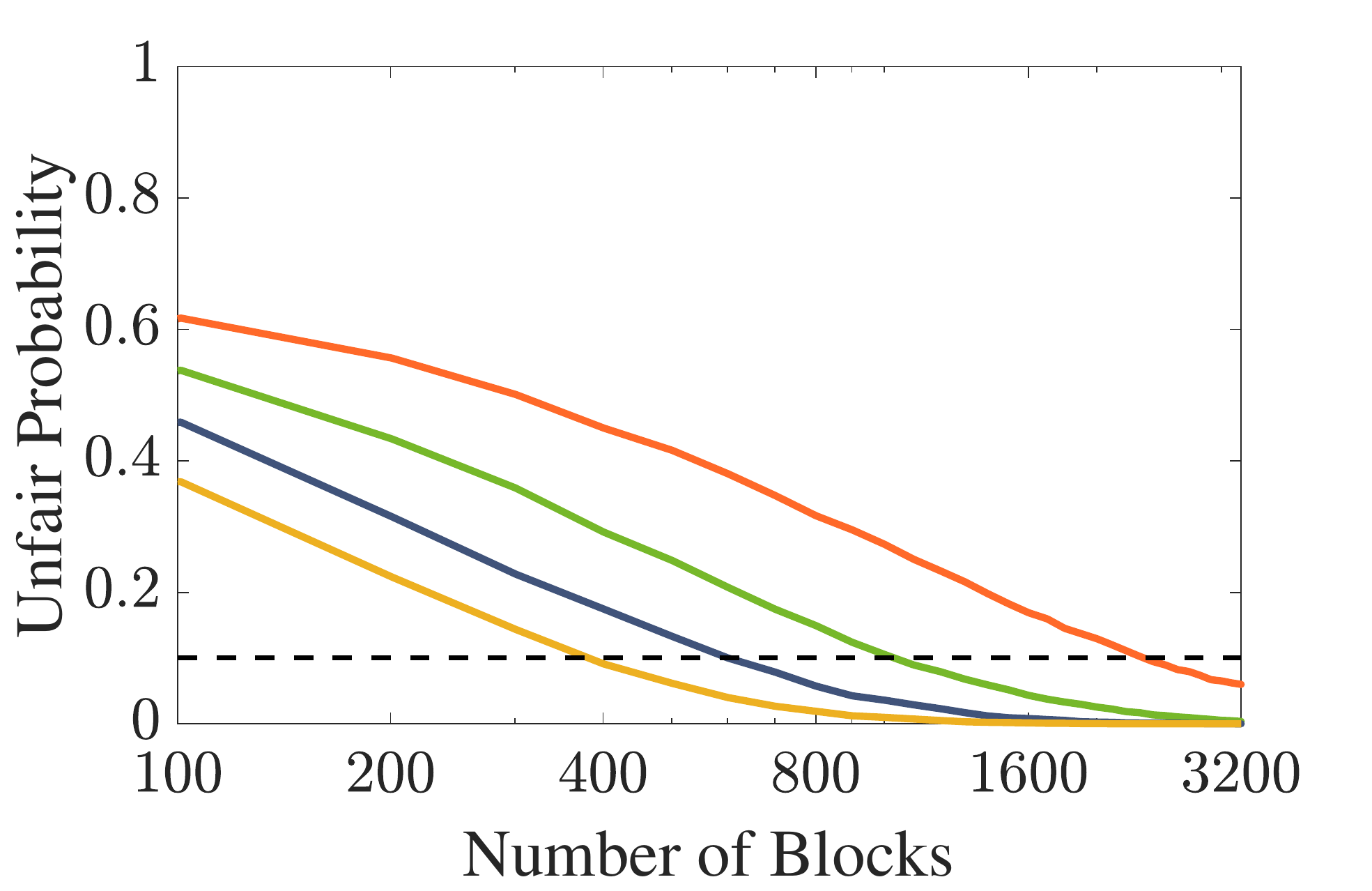}\label{fig:powdifset}}\hfil
	\subfloat[ML-PoS]{\includegraphics[width=0.25\textwidth]{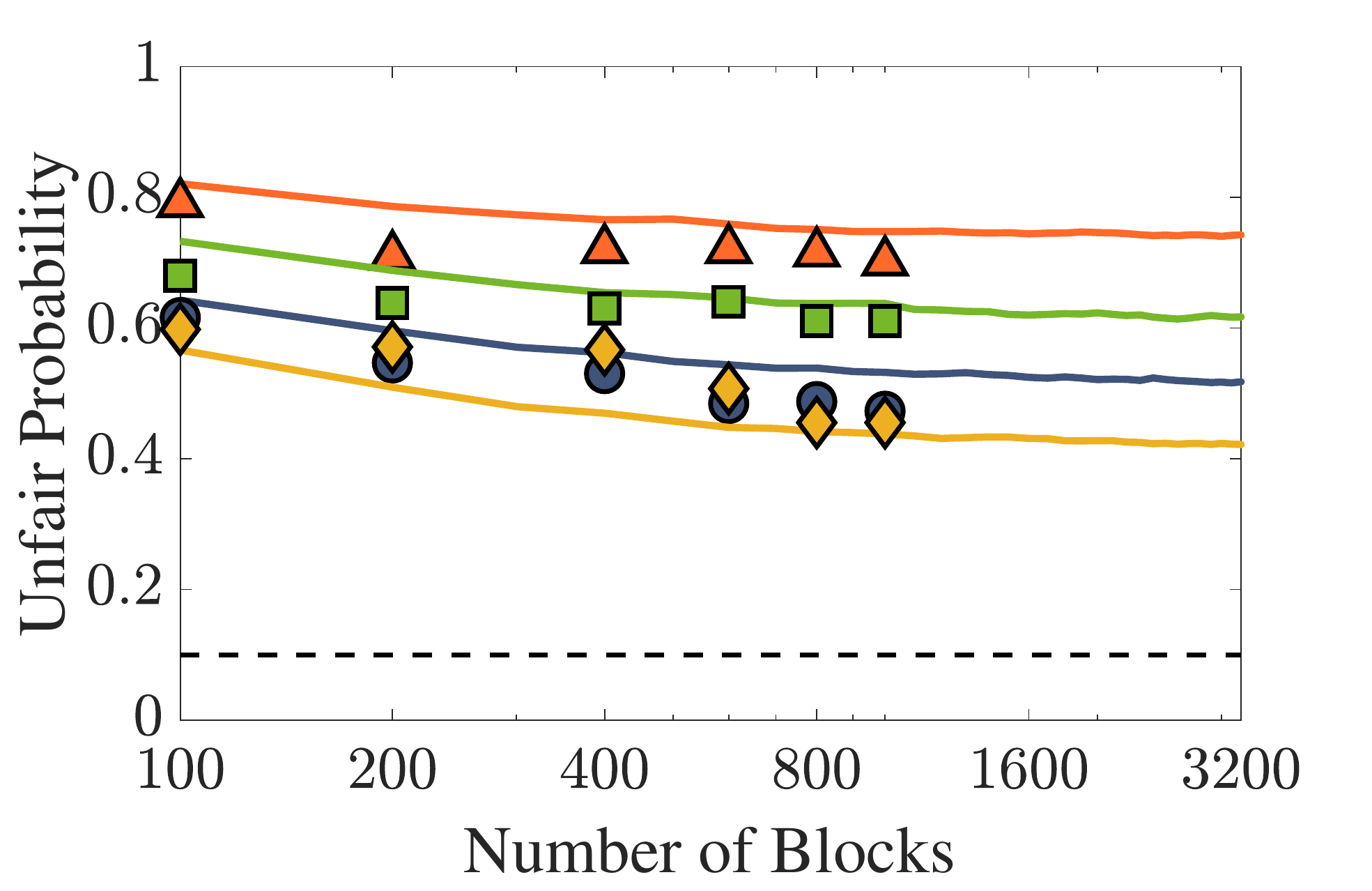}\label{fig:mlposwealth}}\hfil
	\subfloat[SL-PoS]{\includegraphics[width=0.25\textwidth]{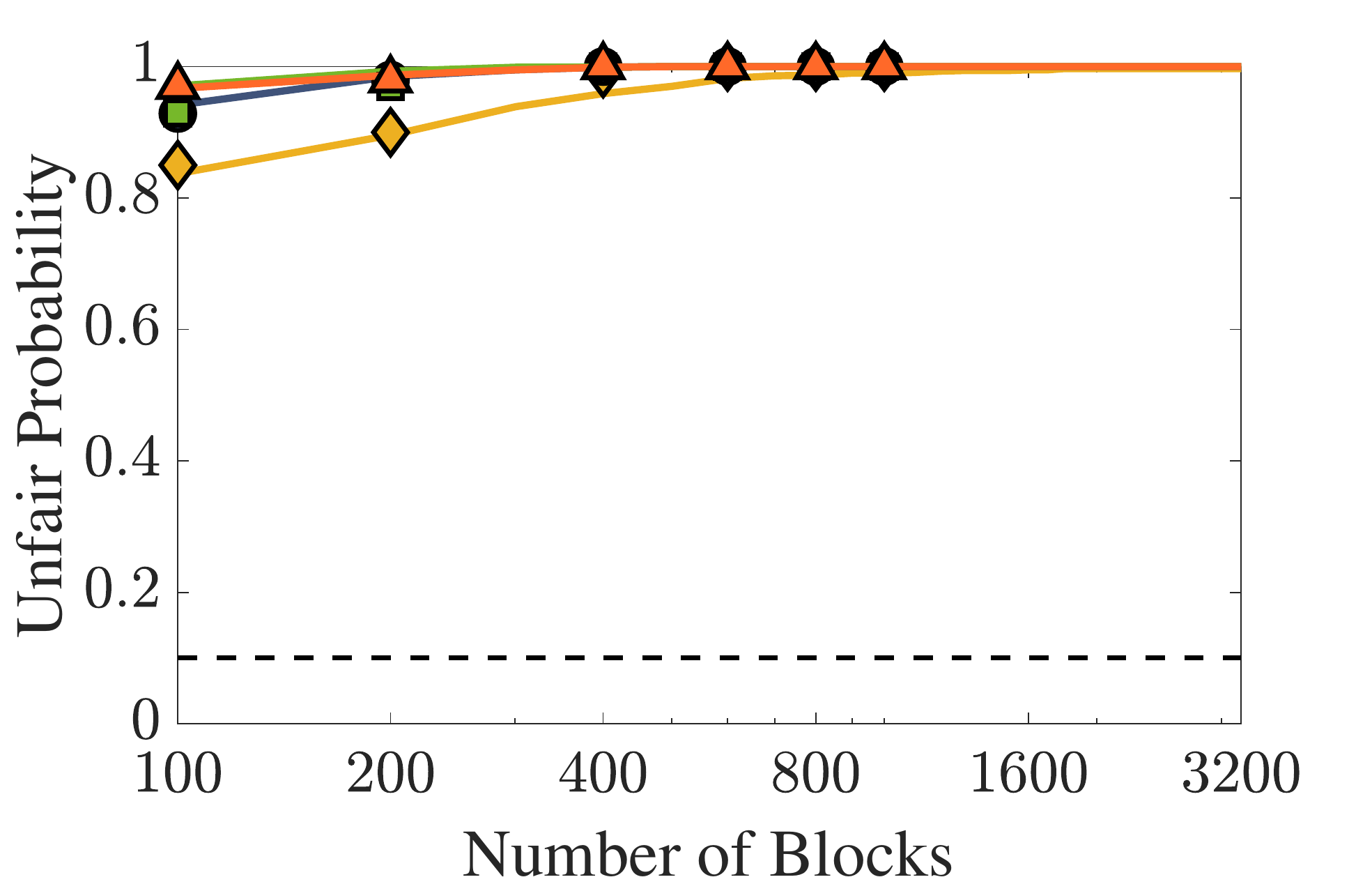}\label{fig:sl-posdifwealth}}\hfil
	\subfloat[C-PoS]{\includegraphics[width=0.25\textwidth]{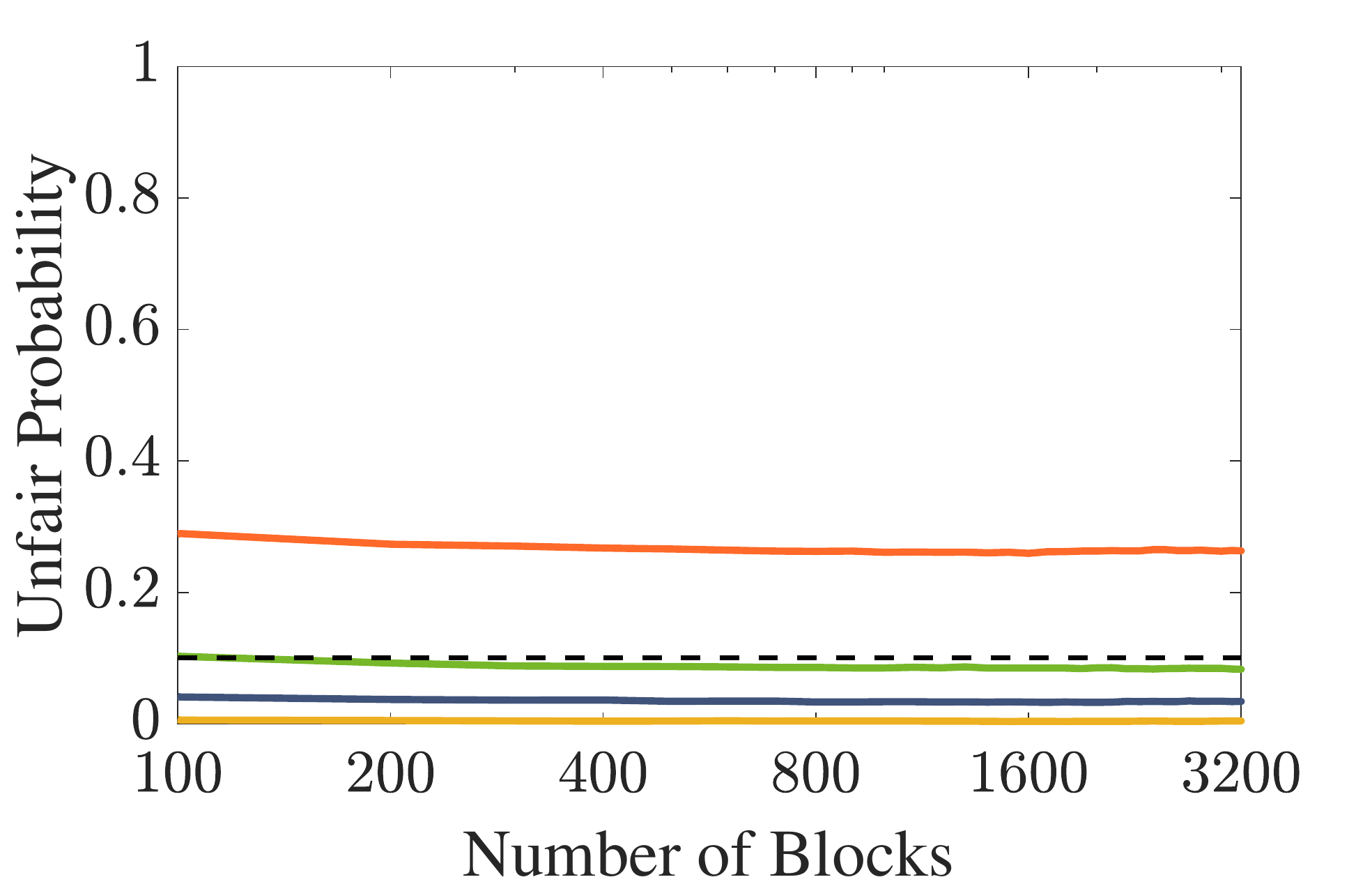}\label{fig:cposdifwealth}}
	\vspace{-3mm}
	\caption{Unfair probabilities for PoW, ML-PoS, SL-PoS and C-PoS under $w=0.01$, $v=0.1$ and different settings of $a$.}
	\label{robustwelath}
	\vspace{-4mm}
\end{figure*}

\figurename~\ref{fig:syssim} illustrates the evolution of $\lambda_A$ with the number of blocks $n$ that miners compete, which can capture both expectational fairness and robust fairness. Initially, miner $A$ controls $a=20\%$ of the total mining power for all the tested cases. For ML-PoS and SL-PoS, the reward of each block is set as $w=0.01$ (which is normalized against the total initial stakes). For C-PoS, the proposer and inflation rewards of each mining epoch are set as $w=0.01$ and $v=0.1$ where the inflation reward is $10$ times of the proposer reward like the settings in Etherum 2.0. In \figurename~\ref{fig:syssim}, the area between two black dash lines is the fair area, \ie~$[(1-\varepsilon)a,(1+\varepsilon)a]$. Meanwhile, the orange line represents the sample average of all the simulation results, and the bottom and top edges of the blue area indicate the 5th and 95th percentiles, respectively. That is, expectational fairness is achieved if the orange line matches the horizontal line with a value of $0.2$, while robust fairness is achieved if the blue area locates within the fair area. Similarly, the green bar shows the result of the real system experiments, where the central mark indicates the mean, and the bottom and top edges indicate the 5th and 95th percentiles. 




\figurename~\ref{fig:powsyssim} reports the evolution of $\lambda_A$ for PoW. As can be seen, the average of both the experiment and the simulation is very close to $0.2$, which confirms that PoW achieves expectational fairness as stated in Theorem~\ref{thm:PoW-fair}. The result also shows that $\lambda_A$ of both the system and the simulation gradually converges to the fair area as $n$ increases. Specifically, when $n < 100$, there are a noticeable fraction of cases locating in the unfair area. On the other hand, when $n > 1{,}000$, almost all cases locate in the fair area, which confirms Theorem~\ref{thm:PoW-robust}. In the Ethereum mining protocol, the average time interval between two blocks is around $15$ seconds and thus the total reward obtained by a miner will concentrate to the expectation with a high probability after around $4.1$ hours.

\figurename~\ref{fig:mlpossyssim} shows the result for ML-PoS. Again, the average of $\lambda_{A}$ of both the experiment and the simulation is close to $0.2$, which demonstrates the expectational fairness of ML-PoS (\ie~Theorem~\ref{thm:ML-PoS-fair}). However, unlike PoW, we find that there are a large number of cases locating in the unfair area even though miners have competed for a great number of blocks, \eg~$n=5{,}000$ or equivalently $8.6$ days in our experiment. That is, for certain settings of block reward $w$, \eg~$w=0.01$, miner $A$ is likely to feel unfair no matter how long $A$ participates the mining game, which is not robustly fair. Recall that Theorem~\ref{thm:ML-PoS-robust} shows that if $1/n+w\leq \frac{2a^2\varepsilon^2}{\ln({2}/{\delta})}$, ML-PoS can achieve an $(\varepsilon,\delta)$-fairness, which is consistent with our observation. However, in our evaluation, $\frac{2a^2\varepsilon^2}{\ln({2}/{\delta})}\approx 0.00027\ll w=0.01$, which does not satisfy the requirement of robust fairness for ML-PoS. 

\figurename~\ref{fig:slpossyssim} plots the evolution of $\lambda_A$ for SL-PoS. We observe that different from the three other protocols, $\lambda_A$ continuously decreases as the mining game proceeds. Specifically, the average of $\lambda_A$ for the first block is $0.2/(2\times 0.8)=12.5\%$ (Section~\ref{subsec:SL-PoS}), and it decreases to $2\%$ quickly after $10^4$ new blocks are proposed (around $9.2$ days on NXT). Furthermore, $\lambda_A$ even approaches $0$ when $n$ reaches $10^{5}$ (around $92$ days on NXT). This indicates that when $a=0.2$ and $w=0.01$, after the mining game operates for a period of time, the poor miner (\ie~miner $A$) will completely lose her stake share and the rich miner (\ie~miner $B$) will monopolize the future generation of new blocks. This phenomenon exhibits a clear unfairness. These observations confirm our analysis in Theorem~\ref{thm:SL-PoS-fair} and Theorem~\ref{thm:SL-PoS-robust}.

\figurename~\ref{fig:cpossim} gives the result for C-PoS where there are $P=32$ shards. We observe that the average of $\lambda_{A}$, as expected by Theorem~\ref{thm:C-PoS-fair}, is almost $0.2$. Compared with ML-PoS, the distribution of $\lambda_{A}$ for C-PoS has a significantly narrower range. In fact, C-PoS is superior to ML-PoS, leveraging the advantages of sharding and inflation reward to reduce the uncertainty of reward allocation, as analyzed in Theorem~\ref{thm:C-PoS-robust}. Therefore, we conclude that C-PoS is the best PoS protocol among the three tested in terms of fairness.


\subsection{Study on Expectational Fairness for SL-PoS}
In \figurename~\ref{fig:syssim}, we show that among the four examined protocols, only SL-PoS does not ensure expectational fairness. In this section, we further study some factors, including initial stake allocation $a$ and block reward $w$, that may affect the expectational fairness for SL-PoS. The result is given in \figurename~\ref{fig:slposfair1}, where markers and lines represent the experimental and simulation results, respectively.

\begin{figure}[!t]
	\centering  
	 \vspace{-3mm}
	\subfloat[Different stake allocation $a$]{\includegraphics[width=0.25\textwidth]{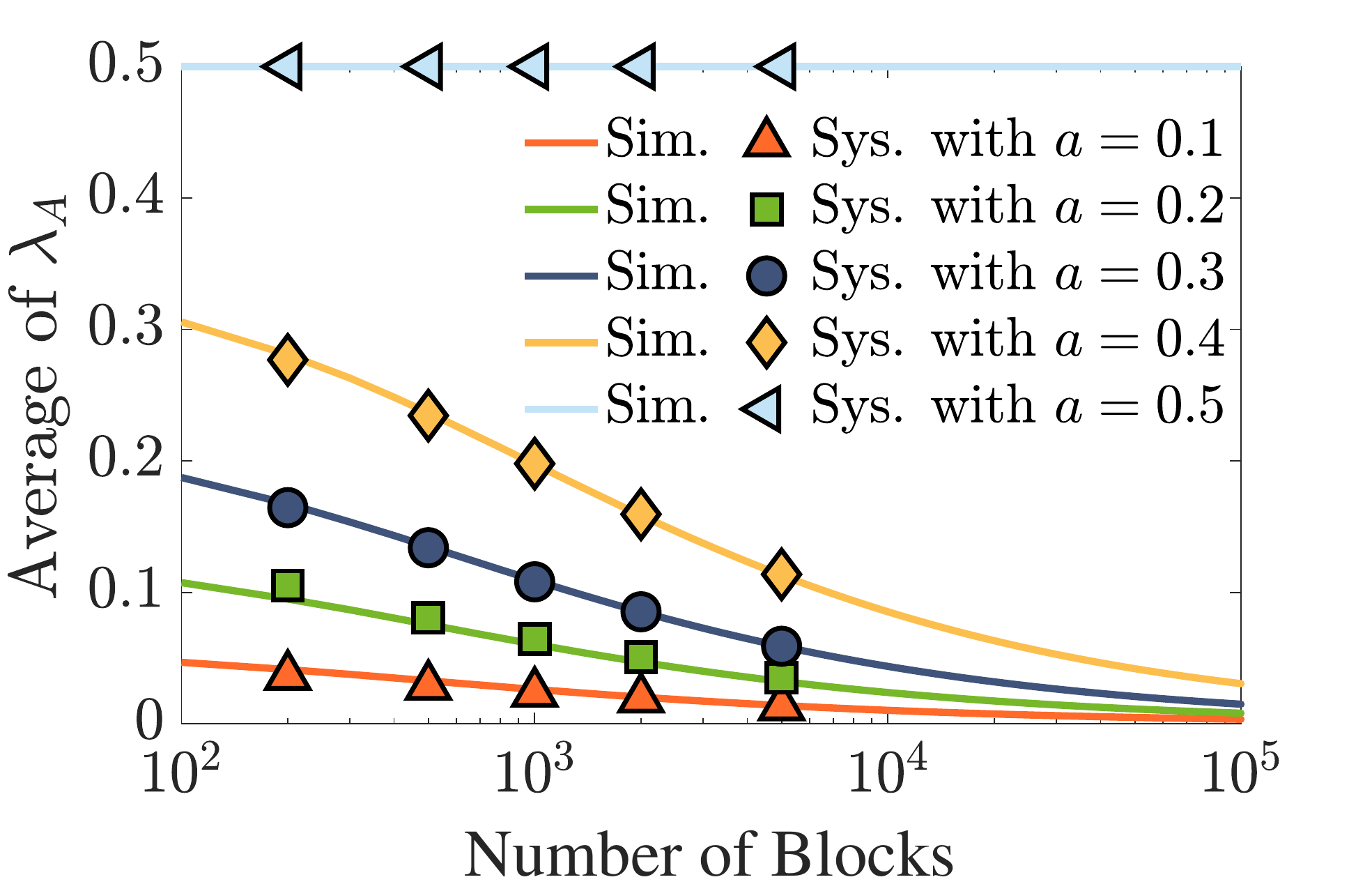}\label{fig:slposdifini}}\hspace*{-0.1in}
	\subfloat[Different block reward $w$]{\includegraphics[width=0.25\textwidth]{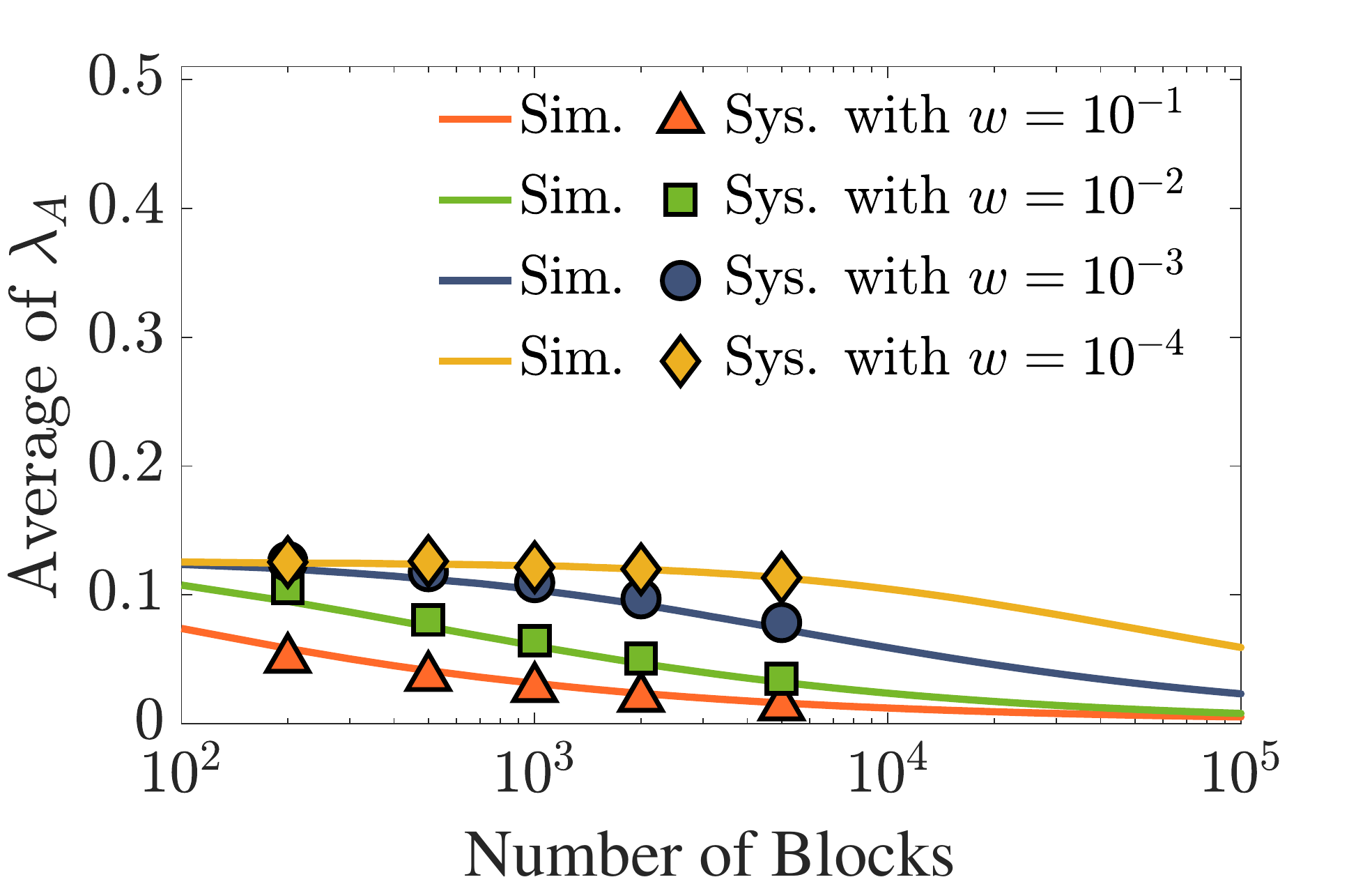}\label{fig:slposdifblock}}
	\vspace{-3mm}
	\caption{Average of reward proportion $\lambda_A$ for SL-PoS.}\label{fig:slposfair1}
	\vspace{-4mm}
\end{figure}

\subsubsection{Impact of Initial Stake Allocation}
\figurename~\ref{fig:slposdifini} reports the average of the reward proportion ${\lambda_A}$ under $w=0.01$ and different staking power $a$ of miner $A$ with values in $\{0.1,0.2,0.3,0.4,0.5\}$. Interestingly, the average reward proportion of $A$ reduces to $0$ for all settings except the one at $a=0.5$. This indicates that no matter how much staking power (once $a<0.5$) miner $A$ controls initially, the miner will own zero staking power finally. We also observe that the average of ${\lambda_A}$ increases along with $a$, which implies that it takes longer time to completely lose competitiveness for the miner with a larger initial staking power.

\subsubsection{Impact of Block Reward}
\figurename~\ref{fig:slposdifblock} shows the average of the reward proportion ${\lambda_A}$ under $a=0.2$ and different block reward $w$ with values in $\{10^{-4},10^{-3},10^{-2},10^{-1}\}$. We observe that the average of $\lambda_A$ decreases along with both (i) the number $n$ of blocks and (ii) block reward $w$. The reason is that if the block reward $w$ is smaller, the fraction of staking power controlled by $A$ reduces slower, so as to the average of $\lambda_{A}$.

\subsection{Study on Robust Fairness}
We further explore the effects of initial resource allocation $a$, block proposer reward $w$ and inflation reward $v$ on robust fairness. We measure the likelihood of $\lambda_{A}$ locating in the unfair area, referred to as \textit{unfair probability}, \ie~$\Pr[\lambda_A<(1-\varepsilon)a \vee \lambda_A>(1+\varepsilon)a]$. This metric reveals that $(\varepsilon,\delta)$-fairness is achieved only if the unfair probability is no more than $\delta$. \figurename~\ref{robustwelath} and \figurename~\ref{robustblock} show the results, where the experimental and simulation results are indicated by markers and lines, respectively. Note that we only show experimental results for ML-PoS and SL-PoS, as repeating PoW experiments $10$ times is insufficient to calculate unfair probability and C-PoS is under the development of Ethereum 2.0.

\begin{figure*}[!ht]
\centering  
\includegraphics[width=0.75\textwidth]{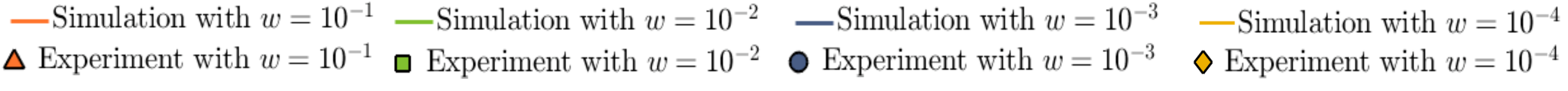}\hfill
\includegraphics[width=0.16\textwidth]{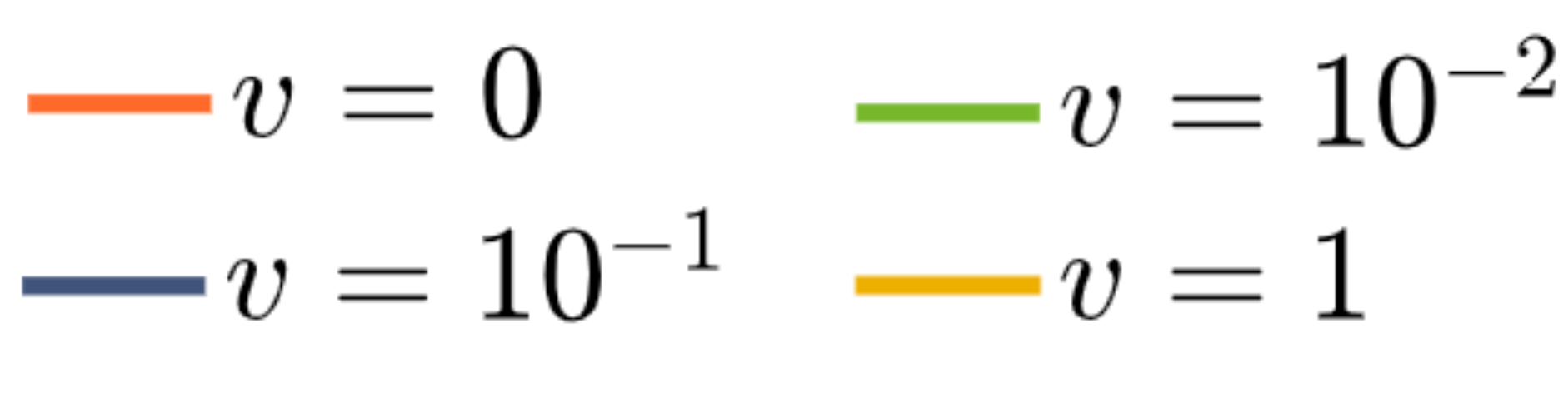}\hspace{0.3in}\vspace{-4mm}\\
\subfloat[ML-PoS]{\includegraphics[width=0.25\textwidth]{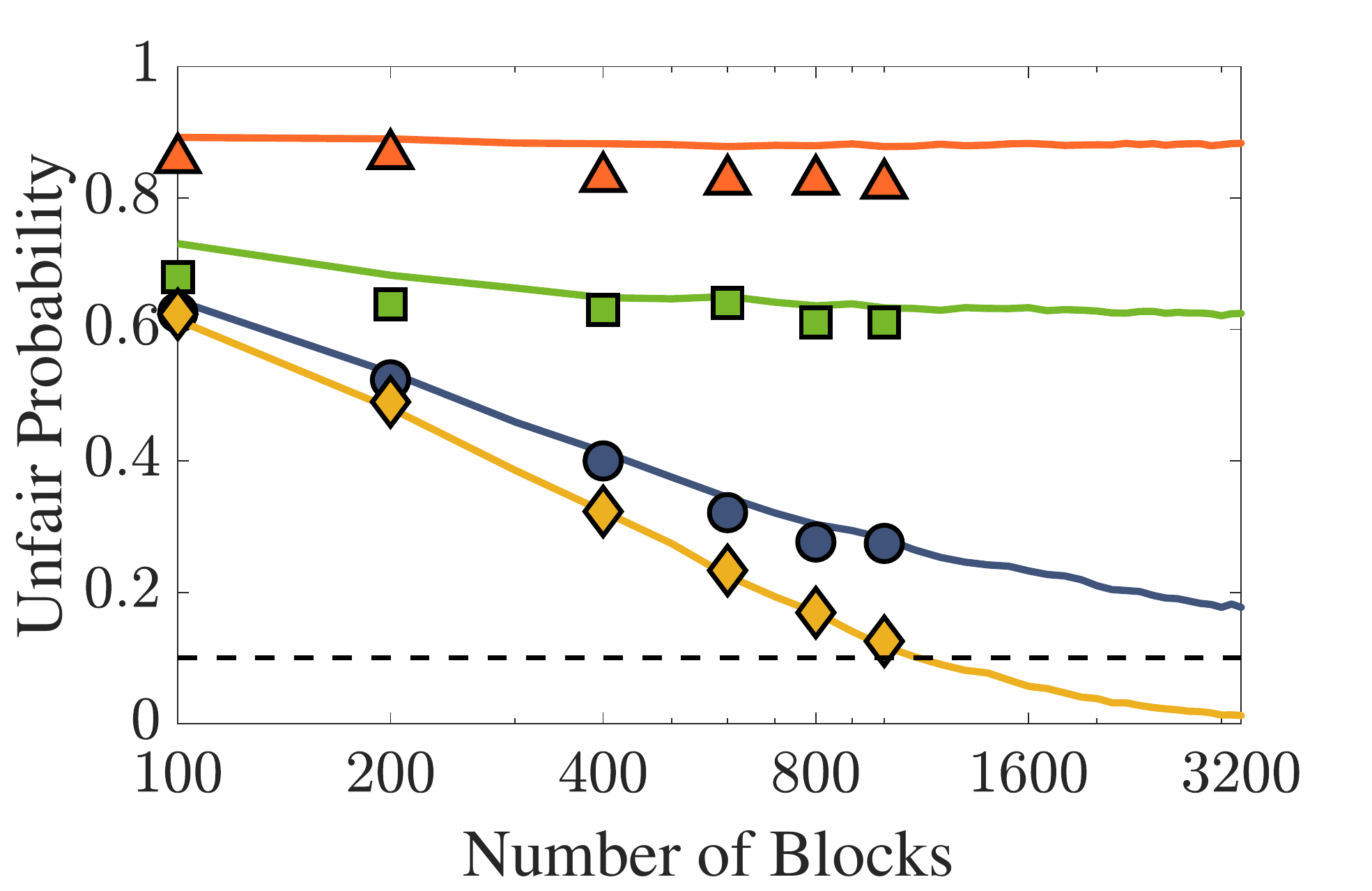}\label{fig:mlposblock}}\hfil
\subfloat[SL-PoS]{\includegraphics[width=0.25\textwidth]{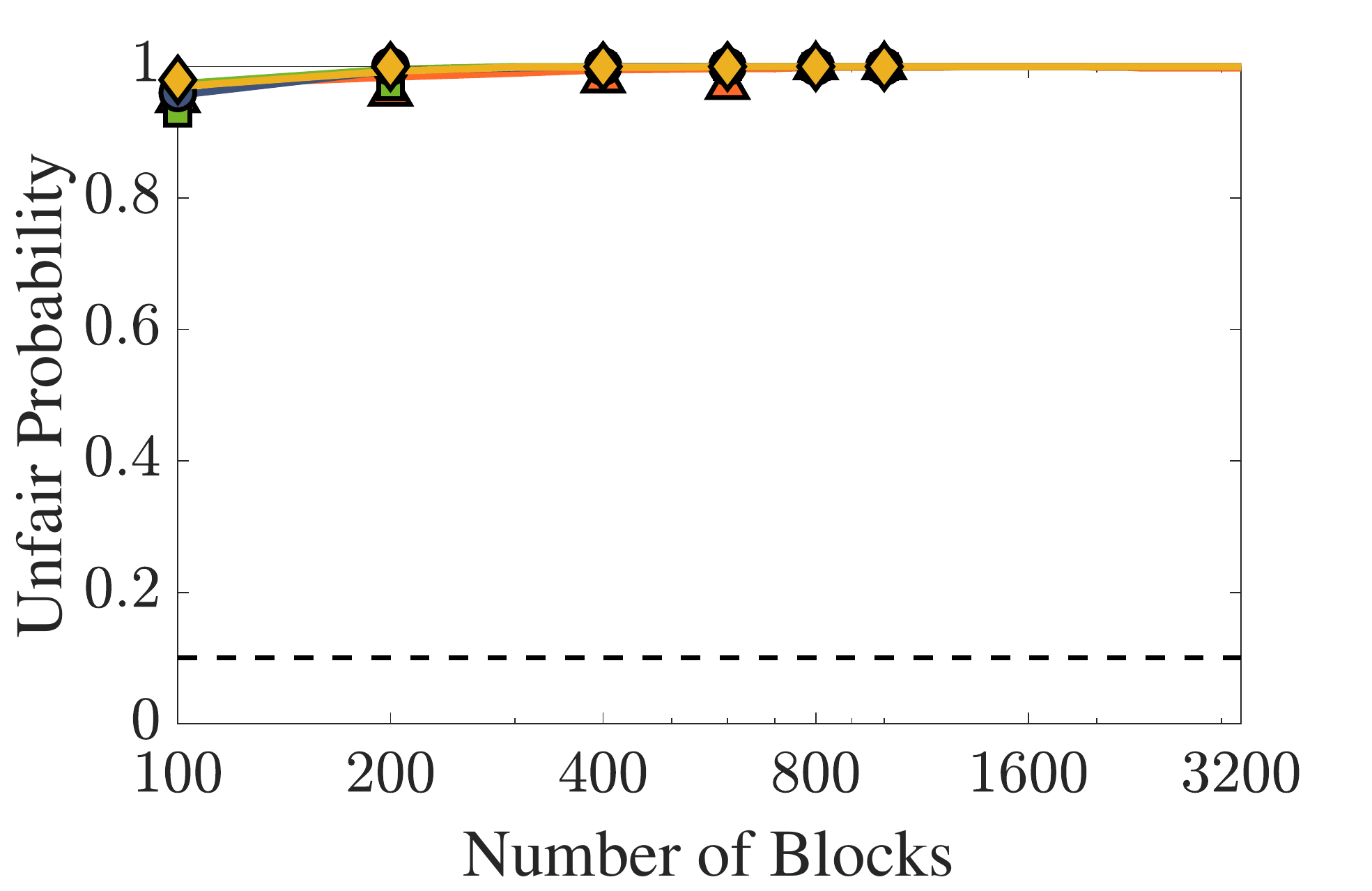}\label{fig:slposdifbloc}}\hfil
\subfloat[C-PoS with $v=0.1$ and different $w$]{\includegraphics[width=0.25\textwidth]{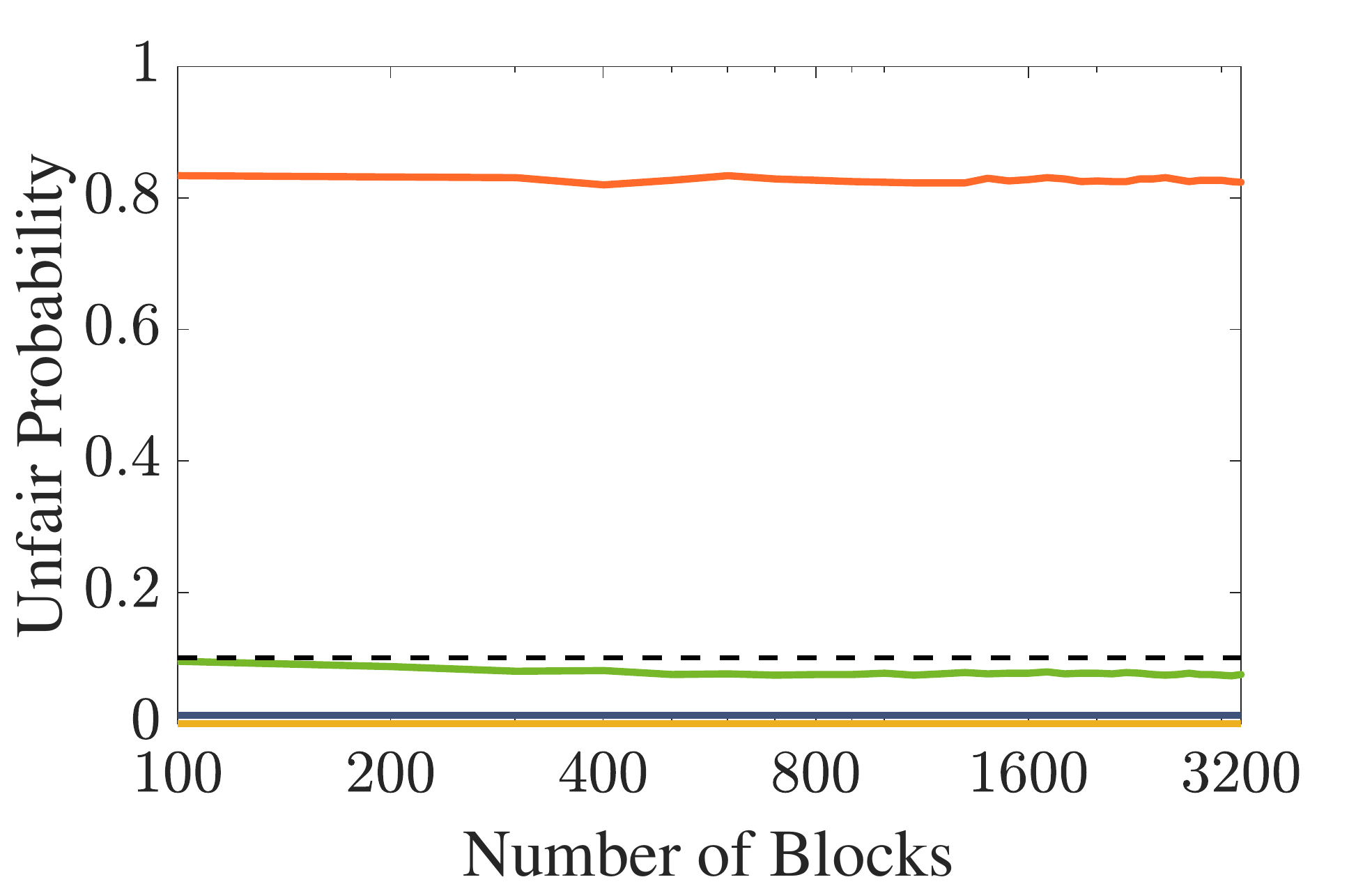}\label{fig:cposdifblock}}\hfil
\subfloat[C-PoS with $w=0.01$ and different $v$]{\includegraphics[width=0.25\textwidth]{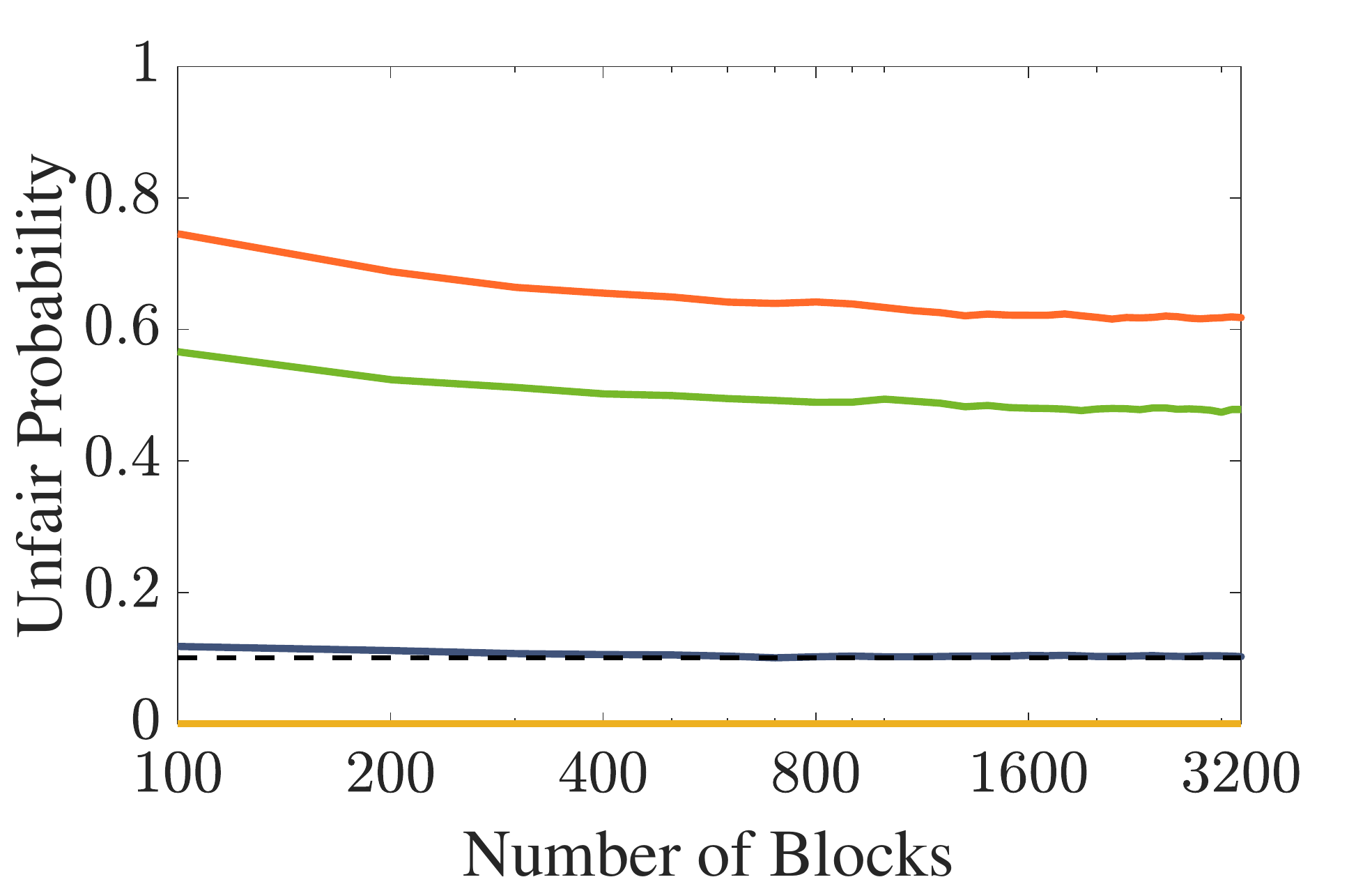}\label{fig:cposinf}}
\vspace{-3mm}
\caption{Unfair probabilities for PoW, ML-PoS, SL-PoS and C-PoS under $a=0.2$ and different settings of $w$ and $v$.}
\label{robustblock}
\vspace{-4mm}
\end{figure*}

\subsubsection{Impact of Initial Resource Allocation}
\figurename~\ref{fig:powdifset} shows the unfair probability for PoW under various $a$, where the black dash line represents the probability threshold $\delta=0.1$ for achieving the $(\varepsilon,\delta)$-fairness. We observe that the unfair probability under all settings reduces along with the number $n$ in general, which again implies that an $(\varepsilon,\delta)$-fairness is always achievable by PoW when $n$ is sufficiently large. We also find that PoW achieves an $(\varepsilon,\delta)$-fairness faster for larger value of $a$. As an example, the number $n$ of blocks required for preserving robust fairness for a medium miner with $a=0.3$ is less than $800$ while that for a tiny miner with $a=0.1$ is more than $2{,}000$. This indicates that if the majority of miners in a PoW-based blockchain system control a small fraction of the total hash power (which is the usual case), it requires a relatively long time period to preserve an $(\varepsilon,\delta)$-fairness for every miner.

\figurename~\ref{fig:mlposwealth} plots the unfair probability for ML-PoS. We observe that at the beginning, the unfair probability decreases along with the number $n$ of blocks. However, when $n$ reaches some thresholds, \eg~$n=1{,}000$, the unfair probability converges to certain constants that are likely to be larger than the threshold of $\delta$. This indicates that a long period of competing time does not suffice the requirement of $(\varepsilon,\delta)$-fairness. Meanwhile, we observe that the unfair probability is smaller for a miner controlling more stakes $a$. This implies that rich miners are more likely to feel fair than poor miners.

\figurename~\ref{fig:sl-posdifwealth} compares the unfair probability for SL-PoS. We observe that the unfair probability of miner $A$ initiates over a wide range, \ie~a tiny miner with $a=0.1$ starts with an unfair probability of $98\%$ and the unfair probability for a large miner with $a=0.4$ is $82\%$. However, in all settings, the unfair probability gradually increases with a larger block number and eventually converges to $100\%$. Moreover, the result shows that the unfair probability of rich miners deteriorates slower than that of poor miners. As an example, the unfair probability of a tiny miner with $a =0.1$ converges to $100\%$ when $n$ reaches $200$ but a large miner with $a=0.4$ turns to $100\%$ unfair when $n$ exceeds $800$.

\figurename~\ref{fig:cposdifwealth} reports the simulation result for C-PoS. In general, C-PoS has similar trends with ML-PoS but the unfair probability of the former is much lower and converges more rapidly. Specifically, for a medium miner with $a=0.3$, the unfair probability of ML-PoS is as high as ${\sim}50\%$ but that of C-PoS is less than $10\%$ such that the $(\varepsilon,\delta)$-fairness is achieved for C-PoS but not for ML-PoS.

\subsubsection{Impact of Block Reward}
We further study how a block reward affects the robust fairness for ML-PoS, SL-PoS and C-PoS. (Note that PoW is insensitive to block reward $w$.) \figurename~\ref{fig:mlposblock} shows the unfair probability for ML-PoS under $a=0.2$ and different block reward settings. The setting of a large block with $w=0.1$, where the block reward is close to the initial stake circulation, suffers from a severe fairness issue. In particular, the unfair probability of miner $A$ is at least $85\%$. The reason is that the mining outcome of the first few blocks will significantly change the distribution of staking power among miners, which in turn heavily affects the mining game subsequently. As a contrary, in the setting of tiny block with $w=10^{-4}$, achieving an $(\varepsilon,\delta)$-fairness for miner $A$ is easy. The reason is that if block reward is much smaller compared with the initial stakes, the earned stakes from the mining game have a negligible contribution to the staking power. In other words, the probability for a miner to propose a new block remains roughly unchanged when time evolves. Therefore, to improve the fairness for ML-PoS, we may set a small reward for each block or release more stakes at the very beginning of the mining game. 

\figurename~\ref{fig:slposdifbloc} illustrates the unfair probability for SL-PoS, which is relatively insensitive to the block reward $w$. In particular, the unfair probabilities for SL-PoS initiate around $95\%$ and then increase to $100\%$ after $200$ blocks for all the settings of $w$ tested. \figurename~\ref{fig:cposdifblock} reports the result for C-PoS when $w$ varies. Again, C-PoS outperforms ML-PoS significantly, though they have similar trends under different settings of $w$. Moreover, we also compare the unfair probability under different settings of the inflation reward $v$ in \figurename~\ref{fig:cposinf}. As can be seen, the unfair probability decreases along with inflation reward $v$. Specifically, the unfair probability under $v=0$ is ${\sim}70\%$ whereas this value sharply reduces to ${\sim}50\%$ under $v=0.01$ and even to ${\sim}10\%$ under $v=0.1$. In intuition, the inflation reward distributed to every miner is completely proportional to their staking power and hence the income uncertainty from the proposer reward is significantly diluted. 
\section{Discussion}\label{sec:discussion}
The previous analysis is based on a two-miner scenario, and this section first discusses how to extend our analysis to a general setting with multiple miners. In addition, the previous analysis reveals that SL-PoS can accomplish neither expectational fairness nor robust fairness while ML-PoS cannot easily achieve robust fairness. This section also discusses some remedies and improvements for these two protocols. Finally, we will discuss more incentives and some practicalities that can benefit from our fairness analysis.
\subsection{Extension to Multiple Miners}\label{dis:exntend-multiple}
\hym{Our analysis above is based on a simple two-miner scenario. In the following, we discuss the fairness in a general setting with multiple miners for the four aforementioned blockchain incentives. 
	
In PoW, ML-PoS and C-PoS, using similar arguments in Section~\ref{sec:preliminaries}, it is trivial to get that the probability of a miner proposing a block is proportional to her computation/staking power regardless of the resource distribution of the other miners. Consequently, according to our analysis in Section~\ref{sec:expectationfariness} and Section~\ref{sec:robustfair}, one can easily verify that the results of both expectational fairness and robust fairness still hold for the three incentives, by considering $B$ as a set of miners. 

On the other hand, such an extension for SL-PoS is non-trivial. In fact, we show that the probability of proposing a block is not proportional to the miner’s staking power in general (unless all miners possess an identical amount of staking power).
\begin{lemma}
\label{lem:C-PoS-multi-player}
In SL-PoS with multiple miners, there exists a miner such that the probability of the miner proposing a block is not proportional to her staking power unless all miners possess an identical amount of staking power. 
\end{lemma}

By Lemma~\ref{lem:C-PoS-multi-player}, analogous to our analysis in Section~\ref{subsec:SL-PoS-EF} and Section~\ref{subsec:SL-PoS-RF}, we can get that neither expectational fairness nor robust fairness is accomplished by SL-PoS when there are multiple miners.

\begin{table}[!t]
	\centering
	\caption{Results for Multi-Miner Game.}\label{tab:multi}
	\vspace{-0.1in}
	\begin{tabular}{llccccccc}
		\toprule
		& \multicolumn{1}{c}{No.\ of Miners}   & PoW      & ML-PoS     & SL-PoS   & C-PoS    \\
		\midrule
		\multirow{5}{*}{\rotatebox{90}{Avg.\ of $\lambda_{A}$}}
		& 2 Miners &  0.20     &  0.20     &   0.00 &  0.20  \\
		& 3 Miners          & 0.20  & 0.20  &0.00 & 0.20  \\
		& 4 Miners  & 0.20 & 0.20   &0.00& 0.20        \\ 
		& 5 Miners            & 0.20 & 0.20 & 0.20  & 0.20          \\ 
		& 10 Miners          & 0.20 & 0.20  & 0.98 & 0.20         \\ 
		\midrule
		\multirow{5}{*}{\rotatebox{90}{Unfair Prob.}} 
		& 2 Miners            &  0     &  0.14     & 1   & 0.08   \\
		& 3 Miners             & 0  & 0.13   & 1& 0.09  \\
		& 4 Miners          & 0 & 0.14  & 1   &  0.08      \\ 
		& 5 Miners         & 0 & 0.15   & 0.98 & 0.08         \\ 
		& 10 Miners      & 0 & 0.13 & 1   &  0.08       \\ 
		\midrule
		\multirow{5}{*}{\rotatebox{90}{Cvg.\ Time}} 
		& 2 Miners       &  1055     &  Never    & Never   &  110  \\
		& 3 Miners   & 1016  & Never  & Never   & 104  \\
		& 4 Miners      & 1087 & Never & Never   & 115      \\ 
		& 5 Miners    & 1010 & Never & Never   & 122          \\ 
		& 10 Miners      & 1030 & Never & Never   & 137        \\
		\bottomrule
	\end{tabular}
	\vspace{-0.1in}
\end{table}
Table~\ref{tab:multi} shows the empirical results via simulations. We compare the average value of $\lambda_A$ and the corresponding unfair probability. We also record the number of blocks when the $(\varepsilon,\delta)$-fairness is achieved. In our simulation, miner $A$ controls $20\%$ of the initial mining resource and the other miners equally share the $80\%$ remaining mining resource. By default, $w=0.01$ and $v=0.1$. We can see that for PoW, ML-PoS and C-PoS, the results of multi-miner are similar to those of two-miner. This indicates that our analysis on these three protocols also holds when more than two players are included in the mining game. However, for SL-PoS, we observe that the average of ${\lambda_A}$ remains $0$ when $2$--$4$ miners are considered in the mining game, and it suddenly increases to $0.2$ and $0.98$ under 5 and 10 miners, respectively. The result implies that miner $A$'s reward depends not only on her staking power but also on the staking distribution of $A$'s competitors. Specifically, only the biggest miner will monopolize the network with a high probability and the rest miners will finally lose their wealth, which extends the conclusion on two-miner to multi-miner. For example, when $5$ miners compete in the network, all miners including miner $A$ have identical $20\%$ stakes initially. Thus, as discussed above, the average income of miner $A$ should be $20\%$ of the total reward. However, when more than $5$ miners are included, miner $A$ controls more stakes than the others and thus the average income of miner $A$ improves dramatically. 
}
\subsection{Treatment for SL-PoS}\label{dis:slpos}
SL-PoS will finally turn to monopolization due to the unfair winning probability for each block. One potential treatment is to adjust the $\mathtt{time}$ function so that the winning probability becomes asymptotically proportional to the staking power. 

Specifically, assume that $X$ and $Y$ denote the hash values of candidate blocks issued by $A$ and $B$, which are uniformly distributed in the range of $[0,2^{256}-1]$ such that $\frac{X}{2^{256}}$ and $\frac{Y}{2^{256}}$ follow uniform distribution $U(0,1)$ asymptotically. To ensure both expectational and robust fairness, the $\mathtt{time}$ function $T(\cdot)$ is required to satisfy $\Pr[T(X;S_A)<T(Y;S_B)]=S_A/(S_A+S_B)$, where $S_A$ and $S_B$ are the staking power of $A$ and $B$. Motivated by the PoW incentive model in Section~\ref{subsec:PoW}, if $T_A=T(X;S_A)$ and $T_B=T(X;S_B)$ follow negative exponential distributions with rate parameters $S_A$ and $S_B$, \ie~the probability density function of $T_A$ is $f(t_A; S_A)=S_A\e^{-S_At_A}$ for $t_A\geq 0$. To achieve this goal, we make use of inverse transform sampling. In particular, consider the cumulative distribution function of $T_A$ as $F(t_A;S_A)=1-\e^{-S_At_A}$. Now, let $F(T_A;S_A)=X$ such that $T_A=\frac{-\ln(1-X)}{S_A}$. We know that if $X$ is uniformly distributed in the range of $[0,1]$, $T_A$ is a random variable following the negative exponential distribution with a rate parameter $S_A$. Applying the same approach for $T_B$, we can obtain that $\Pr[T_A<T_B]=\frac{S_A}{S_A+S_B}$. To conclude, our treatment for SL-PoS is to set $\mathtt{time} = \mathtt{basetime}\cdot \frac{-\ln(1 - {\operatorname{Hash}(\mathtt{pk},...)}/{2^{256}})}{\mathtt{stake}}$.
\hym{We conduct both experiments and simulations on NXT to evaluate our treatment, referred to as FSL-PoS (\ie~fair-single-lottery PoS). \figurename~\ref{fig:improved_nxt_simreal} shows the evolution of $\lambda_A$ along with the number $n$ of blocks that miners compete. In contrast to the original SL-PoS in \figurename~\ref{fig:slpossyssim}, the average income of miner $A$ in FSL-PoS is $20\%$ of the total reward, which confirms the expectational fairness of our treatment. However, we observe that quite a few cases of $\lambda_A$ locate outside the fair area, which indicates that robust fairness is not achieved yet. In what follows, we further discuss how to improve robust fairness.}
\begin{figure}[!t]
	\centering  
	\includegraphics[width=0.33\textwidth]{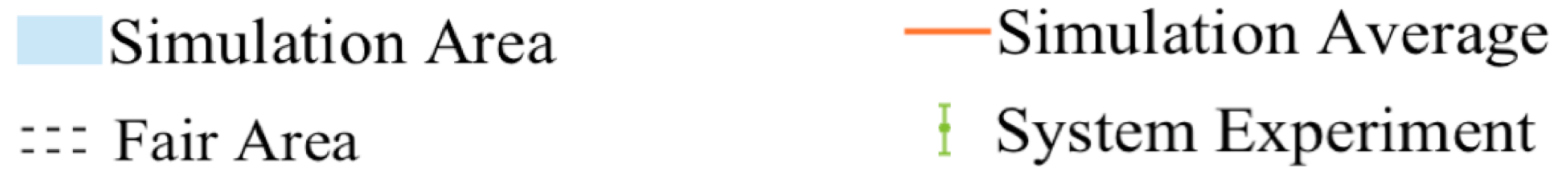}\vspace{-4mm}\\
	\subfloat[FSL-PoS]{\includegraphics[width=0.25\textwidth]{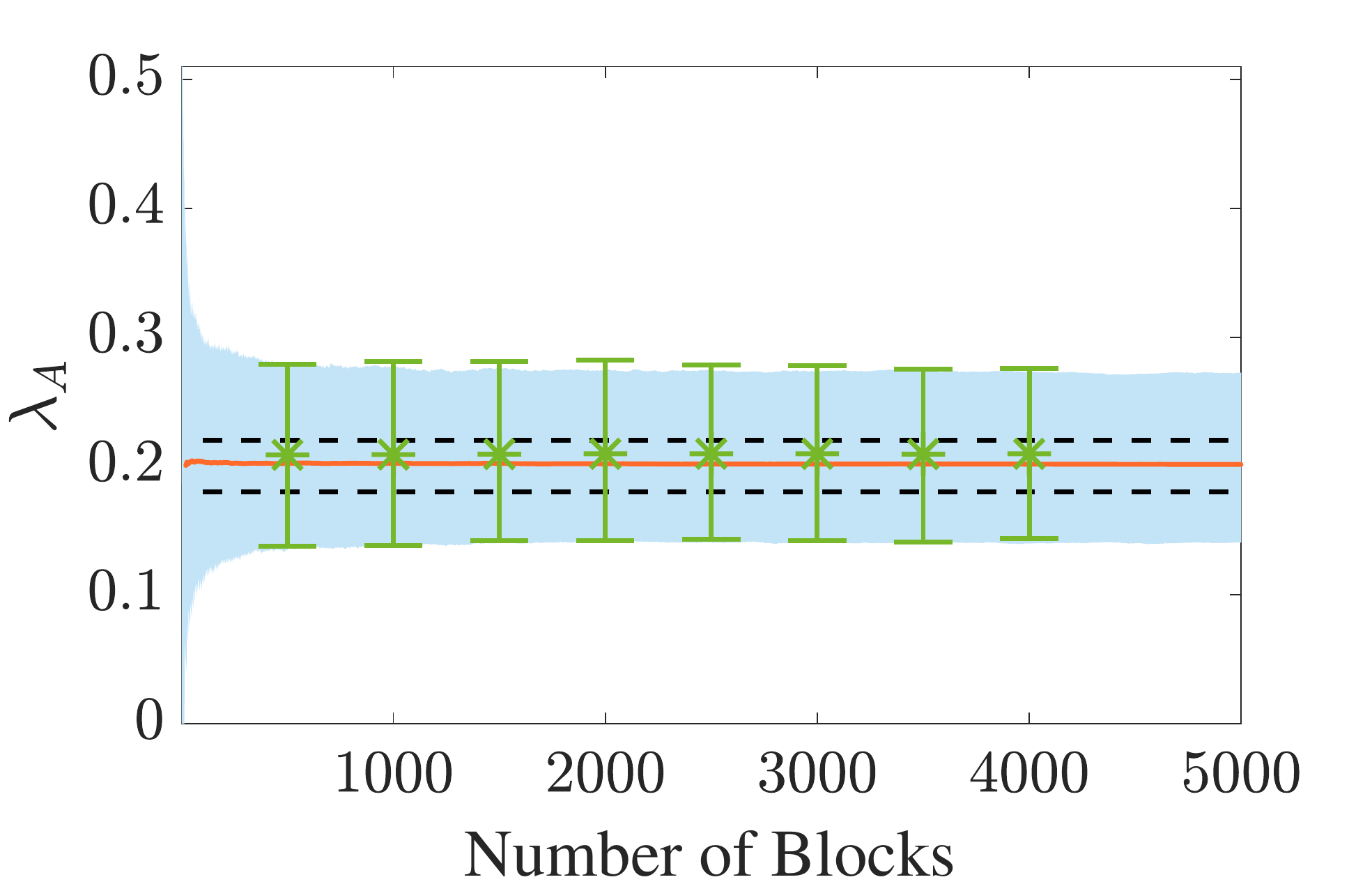}\label{fig:improved_nxt_simreal}}\hspace*{-0.1in}
	\subfloat[FSL-PoS Reward Withholding]{\includegraphics[width=0.25\textwidth]{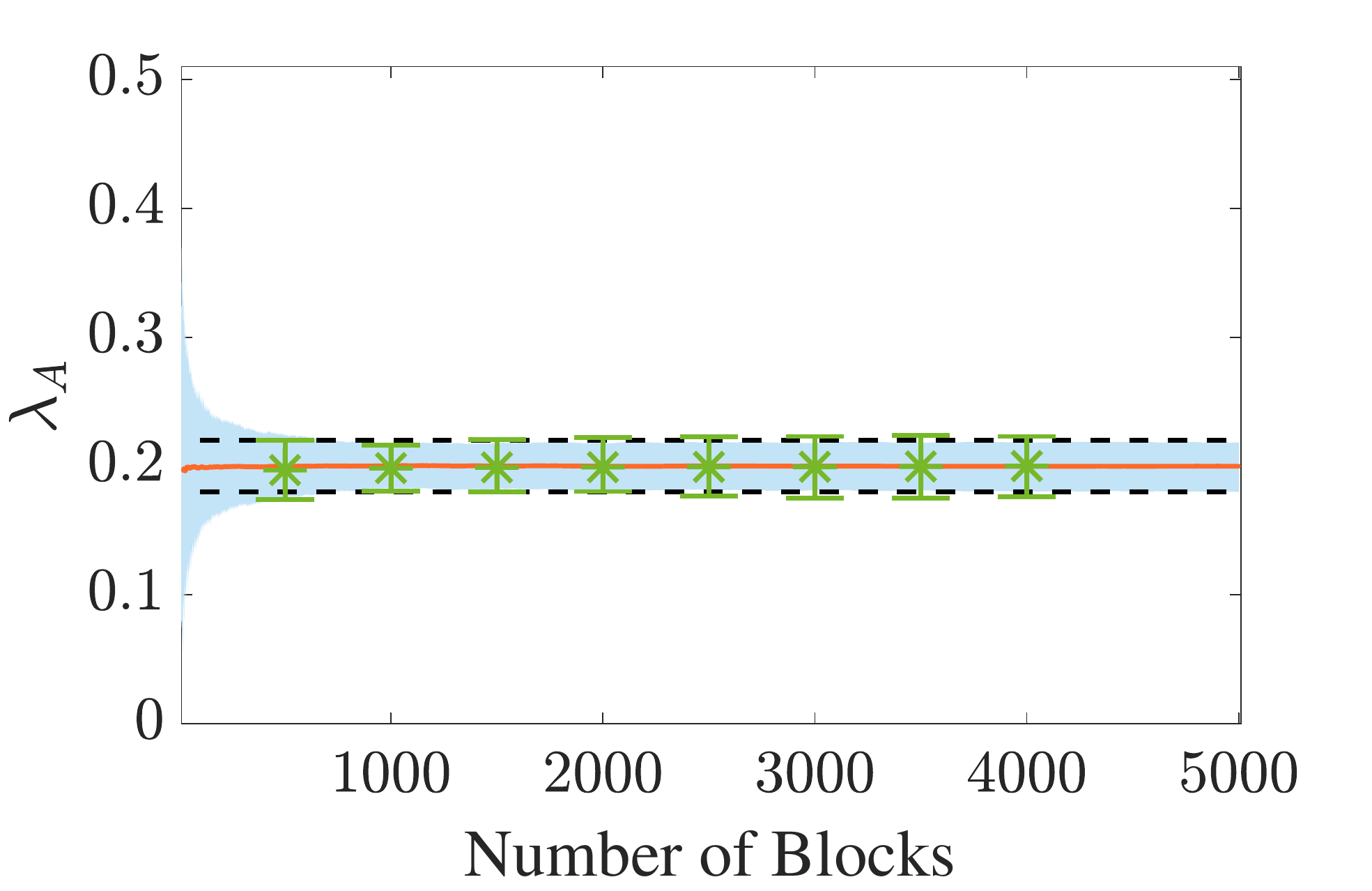}\label{fig:improved_2_nxt}}
	\vspace{-3mm}
	\caption{Evolution of $\lambda_A$ along with the number $n$ of blocks under $a = 0.2$,  $w=0.01$.}\label{fig:improve}
	\vspace{-4mm}
\end{figure}
%
%
\subsection{Improvement for Robust Fairness}\label{dis:impro-robust}
\eat{In this section, we discuss several solutions to improve robust fairness based on our analysis in previous sections. From section \ref{sec:robustfair}, we reveal that two settings significantly affect the fairness: the initial allocation of stakes ($a$ and $b$) and the number of stakes issued in each block ($w$). Therefore, the countermeasures centers around improving initial stakes and/or reducing stake reward in each block. }

\vspace{-1mm}
\spara{Reward Withholding}
\label{dis:reward-withhold}
One potential solution that may improve robust fairness is to withhold block rewards that will take effect periodically. As an example, the block reward will be issued to the proposer immediately but only take effect at the next $1{,}000$-th block, \eg~the reward is issued at the $1{,}024$-th block but takes effect at the $2{,}000$-th block. The incentives obtained by a miner during two successive effective time points should concentrate to the expectation due to the law of large numbers. As a consequence, fairness will be improved. \hym{We perform experiments and simulations that apply reward withholding to FSL-PoS, where the reward takes effect at the block of the next thousand. \figurename~\ref{fig:improved_2_nxt} reports the evolution of $\lambda_A$ when time evolves. Clearly, almost all cases locate in the fair area, which demonstrates the effectiveness of our improvement.}

\spara{Less Block Reward}
As analyzed theoretically and empirically, a small block reward $w$ is in favor of fairness for ML-PoS. However, this action should be carefully performed since less subsidy will reduce miners' motivation. 
In addition, increasing the initial circulation of stakes can indirectly reduce the relative block reward $w$, which will eventually benefit the fairness of ML-PoS. Initial Coin Offering (ICO) and airdrop are two common ways to allocate initial stake circulation. ICO allows a project team to sell a fraction of stakes to investors before the mining competition. Airdrop allocates cryptocurrencies towards community users for free during the early stage of mining.

\subsection{Fairness of More Incentive Protocols}
\label{sec:more-protocols}

The metrics and insights of fairness that we learned from PoW and three PoS protocols can be applied to more incentive protocols. In the following, we discuss six more incentives. 


\spara{NEO and NEO Gas} NEO \cite{Neo} is a PoS based blockchain that adopts a decentralized Byzantine fault tolerance consensus algorithm among authenticated stakers. The stakers compete for rewards depending on the share of \textit{base asset} (\eg~NEO token) that they possess. Different from other PoS protocols, the rewards and transaction fees are paid for by a separate \textit{reward asset} (\eg NEO gas) that does not affect the future mining power.
Therefore, such a PoS incentive works as same as the conventional PoW protocol, which preserves both types of fairness in a long-term mining game. 

\spara{Algorand} Algorand \cite{gilad2017algorand} is a scalable blockchain adopting verifiable random functions and the Byzantine agreement. It just provides inflation rewards to the stakers who possess Algorand in wallet while no proposer reward is released in the mining process. As a result, the stakers will always obtain fair rewards without uncertainty. Despite its fairness, the incentive model has been questioned by criticisms, since consensus participants may lose motivation to maintain the ledger.

\spara{EOS} EOS \cite{EOS} is a delegated PoS protocol based on the practical Byzantine fault tolerance. It achieves consensus among a committee with $21$ elected delegates who propose blocks by turns. Every delegate proposes the same amount of blocks in a consensus round if she is active and honest. As for the incentive, each delegate receives an inflation reward proportional to her stakes (or votes) and a proposer reward which is a constant for everyone regardless of her stake share. Therefore, in general, neither expectational fairness nor robust fairness is achieved in EOS.

\spara{Wave and Vixify} \citet{begicheva2018fair} proposed a variant of PoS, called Wave, on the basis of NXT (\ie~SL-PoS), which improves the $\mathtt{time}$ function in NXT in a way similar to our treatment FSL-PoS. \citet{orlicki2020sequential} proposed Vixify by imitating the Nakamoto consensus in PoS leveraging verifiable random functions and verifiable delay functions. These two protocols ensure that a miner proposes a new block with a probability proportional to her stakes and only provide a proposer reward that will constitute future mining power.
Therefore, analogous to FSL-PoS or ML-PoS, both Wave and Vixify can achieve expectational fairness but do not ensure robust fairness.

\spara{Filecoin} Filecoin \cite{Filecoin} aims to build a storage network where clients upload and retrieve data in a decentralized way. The system utilizes a  Proof-of-Storage-and-Time protocol to ensure the retrievability of stored data. The incentive is based on the miners' contributions on both storage space and pledge stakes, which constitute mining power. Hence, our analysis of PoW and PoS protocols is useful for understanding the fairness of the Filecoin incentive.

\subsection{Practicality of Fairness Analysis}\label{subsec:security}
 \vspace{-1mm}
\spara{Protecting Data Reliability and Integrity}\label{dis:pract-effect}
Transaction processing and data provenance in permission-less blockchains rely on decentralized governance. A fair incentive is an essential component of public ledgers, since if a system is unfair by design, attackers or whale miners may easily accumulate their wealth during a mining game so that the network becomes centralized gradually. Monopoly miners can maliciously rollback transactions and tamper with data by concentrating mining power on launching a $51\%$ attack. Recently, transactions in Ethereum Classic were rollbacked due to the $51\%$ attack, resulting in $1.68$ million dollars of loss \cite{ETC_doublespend}. Therefore, as one of the FAT principles of responsible data science \cite{Getoor_RDS_2019}, fairness protects data reliability and integrity in practice.

\spara{Preventing Mining Pools}
\label{sec:poolstaking}
To reduce the uncertainty of reward, an effective strategy for miners is to join mining pools, which however encourages a centralized network and hence betrays the foundation of blockchain. Arguably, large mining pools are bad since they may concentrate power on launching severe attacks, \eg~$51\%$ attack. This issue can be well addressed by leveraging the concept of robust fairness. In particular, an incentive preserving robust fairness ensures that miners receive stable rewards, \ie~the random outcome of a miner's reward is concentrated to its initial investment with high probability. With such an incentive mechanism, miners will lose motivation to join mining pools.

\spara{Enhancing Security}
\label{dis:sec-persp}
As discussed above, improving (expectational and robust) fairness can prevent miners from monopolizing the network or joining mining pools, which decreases the risk of adversarial control of a blockchain. In addition, there are several malicious attacks directly targeting on incentives so as to obtain an unfair profit, such as selfish mining~\cite{kwon2017selfish,eyal2018majority,nayak2016stubborn}, block withholding~\cite{eyal2015miner,luu2015power} and bribery~\cite{shangaobribe,mccorry2018smart}. Our analysis provides insight into further study of the incentive-based attacks, especially in PoS protocols which are rarely explored due to technical challenges. This will eventually be useful for developing secure blockchains resistant to these attacks.

\section{Related work}\label{sec:relate}
\vspace{-1mm}
\spara{Incentive and Fairness}
To ensure data immutability and security of permission-less blockchains, a fair incentive mechanism is often required. To our knowledge, rare research work studied the fairness of blockchain incentives, though there have been some concerns raised by cryptocurrency communities. \citet{fanti2019compounding} introduced the concept of \textit{equitability} defined as the ratio of the incentive variance to the initial resource variance, which unfortunately cannot answer the fairness concern directly. \citet{rosu2019evolution} analyzed the rich-get-richer phenomenon for ML-PoS using martingale and Dirichlet distribution, and claimed that ML-PoS will not face fairness issue as the fraction of rewards obtained by a miner in expectation is equal to her initial resource share, \ie~expectational fairness in our context. However, we introduce a new concept of robust fairness that can better capture the uncertainty of a mining game, and show that ML-PoS may not achieve robust fairness. Moreover, these studies~\cite{fanti2019compounding,rosu2019evolution} focused on the classical ML-PoS protocol deployed on earlier PoS implementations such as Qtum~\cite{Qtum} and Blackcoin~\cite{Blackcoin}, whereas our analysis, in addition to ML-PoS, covers more state-of-the-art implementations including SL-PoS by NXT \cite{NXT} and C-PoS by Ethereum 2.0 \cite{ETH20}. \citet{pass2017fruitchains} designed a fair protocol in similar spirit of Nakamoto's PoW protocol. Different from their work that targeted at protocol design, we analyze the fairness of blockchain incentives for several popular protocols, including PoW, ML-PoS, SL-PoS and C-PoS. 

\eat{However, the equitatbility definition is not only a lack of intuition but also incomparable among different systems. We focus on a robust fairness that can better characterize the relation between resource allocation and reward distribution.} 

The incentives of permission-less blockchains have attracted broader interests from researchers. The attacks on blockchain incentives may result in resource accumulation and further increase the risk of transaction tampering. \citet{Kwon2019BitcoinVB} discussed the movement of miners when mining rigs are applicable on two PoW networks. \citet{tsabary2018gap} and \citet{carlsten2016instability} found that miners may periodically suspend their mining rigs if no block reward is provided by the Bitcoin protocol. Some work studied the attacks on incentives, including selfish mining~\cite{kwon2017selfish,feng2019selfish,eyal2018majority,sapirshtein2016optimal,nayak2016stubborn}, block withholding~\cite{eyal2015miner,luu2015power} and bribery~\cite{shangaobribe,mccorry2018smart}. Our work is from the perspective of fairness that complements these existing studies on blockchain incentives. 

\spara{Transaction Processing}
Designing a transaction processing pipeline with high performance under large scale while ensuring security has been a major research topic \cite{cohen2020reasoning}. \citet{zakhary2019atomic} proposed an atomic cross-chain commitment for permission-less ledgers which ensures an all-or-nothing atomicity property. \citet{herlihy2019cross} extended such an atomicity to the cross-chain deal which can be applied to more types of adversarial commerce. \citet{maiyya2019unifying} integrated fault tolerant replication into atomic commitment for cloud data management. \citet{tao2020sharding} adopted a dynamic sharding algorithm on smart contracts to avoid empty blocks and waste of energy. \citet{amiri2019caper} adopted a directed acyclic graph on permissioned blockchain which supports both confidential transactions and cross-application transactions. In addition, some benchmark evaluations studied the transaction throughput and network latency of various blockchain systems \cite{dinh2018untangling,dinh2017blockbench}, the performance of blockchain index structures \cite{yue2020analysis}, and the performance of memory intensive PoW hash functions \cite{feng2020evaluating}. The security of transaction processing relies on the decentralization of the resource, which is heavily affected by the fairness of incentives, \eg~a $51\%$ attack is likely to occur if the rich get richer. Our work evaluates the fairness of various commonly used incentives and provides insights into blockchain designs to ensure reliable data.

\vspace{-0.5mm}
\spara{Blockchain-as-a-Database}
Blockchain, as a distributed database, becomes popular for various applications.
vChain~\cite{wang2020vchain,xu2019vchain} and GEM$^2$~\cite{zhang2019gem} applied an authenticated data structure to blockchain to ensure query integrity. Merkle$^{inv}$ and Chameleon$^{inv}$~\cite{zhangauthenticated} further reduced the maintenance cost of data authentication on hybridstorage blockchains by leveraging cryptographic proof and chameleon commitment. ResilientDB~\cite{gupta2020resilientdb} utilized a network-topology-aware consensus algorithm to achieve both lower communication latency and network decentralization. FalconDB~\cite{peng2020falcondb} adopted database servers with verification interfaces accessible to clients and stored the digests for query/update authentications on a blockchain to enable efficient and secure collaboration. \citet{buchnik13fireledger} proposed FireLedger, a new communication frugal optimistic permissioned blockchain protocol, to improve throughput. \citet{abadi2020anylog} introduced AnyLog, a decentralized data sharing and publishing platform for IoT data. \citet{ruan2019fine} developed simple interfaces that support smart contracts based provenance information. \citet{qi2020bft} improved storage scalability for blockchain systems by integrating erasure coding. \citet{ruan2020transactional} enhanced the execute-order-validate architecture inspired by the optimistic concurrency control in modern databases. \citet{nawab2019blockplane} designed a middleware and communication infrastructure to ensure byzantine fault-tolerance in datacenter. \citet{amiri2020seemore} leveraged a hybrid state machine replication protocol that avoids crash and malicious failures in cloud environment. Blockchain database usually requires incentives to attract participants, though it is not the main focus of the aforementioned studies. Our work provides insights into incentive designs to further expand the applicability of these blockchain databases.

\spara{Polya Urn Process}
Our analysis utilizes some useful tools, including Azuma inequality \cite{azuma1967weighted} for martingales \cite{doob1953stochastic} and stochastic approximation~\cite{robbins1951stochastic,renlund2010generalized}. In particular, we use Doob's martingale and Azuma's inequality to derive the tail probability on concentration. Moreover, the mining process of PoS is related to the (nonlinear) generalized P\'{o}lya urn~\cite{collevecchio2013preferential,renlund2010generalized,laruelle2019nonlinear,arthur1987non}. For example, the mining process of ML-PoS can be modeled by a classical P\'{o}lya urn \cite{mahmoud2008polya}, where the fraction $ \lambda_A$ of blocks proposed by miner $A$ will converge to a beta distribution almost surely. In addition, to solve the asymptotic convergence for the generalized P\'{o}lya urn, various methods are proposed, including stochastic approximation \cite{kaniovski1995non}, brownian motion embedding \cite{collevecchio2013preferential} and exponential embedding \cite{drinea2002balls,mitzenmacher2004scaling}. In this paper, we apply stochastic approximation to SL-PoS, which shows that $ \lambda_A$ will converge to either $1$ or $0$ no matter how much initial staking power miner $A$ possesses. Using these approaches, our analysis may be extended to more complicated scenarios, \eg~with malicious attacks and games.


\section{Conclusion}
\label{sec:con}
We study the fairness of incentives for several blockchain protocols, including PoW, ML-PoS, SL-PoS and C-PoS. We define two types of fairness, including expectational fairness and robust fairness. Our results show that all the protocols except SL-PoS can preserve expectational fairness. We also find that robust fairness is always achievable for PoW as long as the mining process runs for a sufficiently long time. Meanwhile, ML-PoS is difficult to achieve robust fairness while C-PoS can more easily achieve robust fairness thanks to inflation reward and sharding. Unfortunately, SL-PoS will finally turn to monopolization no matter how the initial stakes distribute, which never achieves robust fairness. Both real system experiments and numerical simulations are carried out to demonstrate our analysis. We provide some insights, \eg~increasing inflation reward and reducing proposer reward, to shed light on future study of blockchain incentives. For future work, we aim to take into account malicious attacks on incentives that can change reward distribution so that more fairness issues will be raised.

\begin{acks}
    We are grateful to Michel van Kessel from Blackcoin, Wenbin Zhong from Qtum, Prysmatic Labs and NXT community for their technical supports. This research is supported by \grantsponsor{R-252-000-A27-490}{Singapore National Research Foundation}{} under grant~\grantnum{R-252-000-A27-490}{R-252-000-A27-490}, and by \grantsponsor{12201520}{HK-RGC GRF}{} projects \grantnum{12201520}{12201520} and \grantnum{12200819}{12200819}.
\end{acks}

\balance
\bibliographystyle{ACM-Reference-Format}
\bibliography{reference}

\appendix
\section{Missing Proofs}\label{sec:appendix-proof}
\begin{proof}[Proof of Theorem~\ref{thm:ML-PoS-fair}]
	Let $X_i\in\{0,1\}$ be a binary random variable indicating whether $A$ is the proposer for the $i$-th block. Let $S_i$ be the total staking power possessed by $A$ after $i$ blocks, \eg~$S_0=a$. Then, it is easy to know that $X_i$ follows Bernoulli distribution with success probability $\frac{S_{i-1}}{1+w(i-1)}$, as the total staking power of all miners for competing the $i$-th block is $1+w(i-1)$. Thus, we have
	\begin{equation*}
		S_{i+1} = S_{i} + wX_{i+1}.
	\end{equation*}
	Taking expectation conditioned on $S_i$ gives
	\begin{equation*}
		\E[S_{i+1}\mid S_i] = S_{i} + \frac{wS_{i}}{1+wi}.
	\end{equation*}
	As a result, we have
	\begin{equation*}
		\E[S_{i+1}]=\E[\E[S_{i+1}\mid S_i]] = \frac{1+w(i+1)}{1+wi}\cdot\E[S_{i}] .
	\end{equation*}
	Recursively, we can get that
	\begin{equation*}
		\E[S_{i}]= S_0\prod_{k=0}^{i-1}\frac{1+w(k+1)}{1+wk}=a(1+wi).
	\end{equation*}
	Therefore, $\E[\lambda_A] = \frac{\E[S_{n}]  - a}{wn} = a$, which concludes the theorem.
\end{proof}

\begin{proof}[Proof of Theorem~\ref{thm:SL-PoS-fair}]
	Again, let $X_i\in\{0,1\}$ be a binary random variable indicating whether $A$ is the proposer for the $i$-th block. Consider that $a\leq b$ such that $\Pr[X_1=0]=1-\frac{a}{2b}$ and $\Pr[X_1=1]=\frac{a}{2b}$. Thus, $\E[X_1]=\frac{a}{2b}<a$ unless $a=b$, which shows that the expected reward of the first block for miner $A$ is unfair in general. 
	
	Next, we show that even an infinity number of blocks are proposed by $A$ and $B$, there exists some $a$ such that $\E[\lambda_{A}]\neq a$, where $\lambda_A =\lim_{n\to \infty}\frac{1}{n} \sum_{i=1}^nX_i$. We prove it by contradiction and assume that $\E[\lambda_{A}]=a$ for all $a$. We observe that
	\begin{equation*}
		\E[n\lambda_A] =\E[n\lambda_A\mid X_1=0] \Pr[X_1=0] +\E[n\lambda_A\mid X_1=1] \Pr[X_1=1].
	\end{equation*}
	By our assumption, $\E[n\lambda_A-X_1\mid X_1] = (n-1)\cdot \frac{a+wX_1}{1+w}$, since miner $A$ possesses a fraction $\frac{a+wX_1}{1+w}$ of total staking power after the outcome of the first block is observed. Thus,
	\begin{equation*}
		\E[n\lambda_A] = \frac{a(n-1)}{1+w}\cdot \Pr[X_1=0] +\Big(1+\frac{(a+w)(n-1)}{1+w}\Big)\cdot \Pr[X_1=1].
	\end{equation*}
	As a result,
	\begin{equation*}
		na =\frac{a(n-1)}{1+w}\cdot\Big(1-\frac{a}{2b}\Big) +\Big(1+\frac{(a+w)(n-1)}{1+w}\Big)\cdot \frac{a}{2b}.
	\end{equation*}
	Rearranging it yields $a(b-a)(nw+1)=0$. This shows a contradiction when $0<a<b$ and hence the theorem is proved.
	%
\end{proof}

\begin{proof}[Proof of Theorem~\ref{thm:C-PoS-fair}]
	The proof is analogous to that of Theorem~\ref{thm:ML-PoS-fair}. Let $Y_i$ be the random variable representing the number of shard proposers assigned to miner $A$ at epoch $i$. Let $S_i$ be the total staking power possessed by $A$ after epoch $i$, \eg~$S_0=a$. Then, it is easy to know that $Y_i\sim\operatorname{Bin}\Big(P,\frac{S_{i-1}}{1+(w+v)(i-1)}\Big)$, as the total staking power of all miners at the beginning of epoch $i$ is $1+(w+v)(i-1)$. Thus, 
	\begin{equation*}
		S_{i+1} = S_{i} + \frac{wY_{i+1}}{P}+ \frac{vS_{i}}{1+(w+v)i}.
	\end{equation*}
	Taking expectation conditioned on $S_i$ gives
	\begin{equation*}
		\E[S_{i+1}\mid S_i] = S_{i} + \frac{wS_{i}}{1+(w+v)i}+ \frac{vS_{i}}{1+(w+v)i}.
	\end{equation*}
	As a result, we have
	\begin{equation*}
		\E[S_{i+1}]=\E[\E[S_{i+1}\mid S_i]] = \frac{1+(w+v)(i+1)}{1+(w+v)i}\cdot\E[S_{i}] .
	\end{equation*}
	Recursively, we can get that
	\begin{equation*}
		\E[S_{i}]= S_0\prod_{k=0}^{i-1}\frac{1+(w+v)(k+1)}{1+(w+v)k}=a\Big(1+(w+v)i\Big).
	\end{equation*}
	Therefore, $\E[\lambda_A] = \frac{\E[S_{n}]  - a}{(w+v)n} = a$, which concludes the theorem.
	%
	%
\end{proof}

\begin{proof}[Proof of Theorem~\ref{thm:PoW-robust}]
	According to Hoeffding inequality \cite{hoeffding1994probability}, we know that 
	\begin{equation*}
		\Pr\big[(1-\varepsilon)a \le \lambda_A \le (1+\varepsilon)a\big]\ge 1 - 2\e^{-2na^2\varepsilon^2}.
	\end{equation*}
	Thus, if $n \ge \frac{\ln(\frac{2}{\delta})}{2a^2\varepsilon^2}$ such that $1 - 2\e^{-2na^2\varepsilon^2}\geq 1 - \delta$, an $(\varepsilon,\delta)$-fairness is preserved.
\end{proof}

\begin{proof}[Proof of Theorem~\ref{thm:ML-PoS-robust}]
	Let $X_i\in\{0,1\}$ be a binary random variable indicating whether $A$ is the proposer for the $i$-th block. Let $S_i:= a + w\sum_{j = 1}^i X_j$ be the number of stakes possessed by $A$ after $i$ blocks. We define $M_i:=\E[S_n\mid X_1,X_2,\dotsc,X_{i}]$ as the expectation of $S_n$ conditioned on $X_1,X_2,\dotsc,X_{i}$. In particular, $M_0=\E[S_n]$ and $M_n=S_n$. Thus, 
	\begin{align*}
		&\E[M_i\mid X_1,X_2,\dotsc,X_{i-1}]\\
		&= \E[\E[S_n\mid X_1,X_2,\dotsc,X_{i}] \mid X_1,X_2,\dotsc,X_{i-1}]\\
		&=\E[S_n\mid X_1,X_2,\dotsc,X_{i-1}]
		=M_{i-1},
	\end{align*}
	which indicates that $M_0,M_1,\dotsc,M_n$ are martingales \cite{doob1953stochastic}.
	In addition, after observing $X_1,X_2,\dotsc,X_i$, miners $A$ and $B$ possess $S_i$ and $1+iw-S_i$ stakes, respectively. Since then, the mining game becomes that $A$ possesses a fraction $\frac{S_i}{1+iw}$ of staking power to compete $(n-i)$ blocks. 
	Similar to the proof of Theorem~\ref{thm:ML-PoS-fair}, we can get that
	\begin{equation*}
		M_i=\E[S_n\mid S_i]=S_i+\frac{(n-i)wS_i}{1+iw}=\frac{1+nw}{1+iw}\cdot S_i.
	\end{equation*}
	\eat{
		\begin{align*}
			M_i
			&=S_i+\E[S_n-S_i\mid X_1,X_2,\dotsc,X_i]\\
			&=S_i+\frac{(n-i)(\frac{a}{w}+S_i)}{\frac{a}{w}+S_i+\frac{1-a}{w}+i-S_i}\\
			&=\frac{(n-i)a+(1+nw)S_i}{1+iw}.
		\end{align*}
	}
	Furthermore,
	\begin{align*}
		M_i-M_{i-1}
		&=\frac{1+nw}{1+iw}\cdot S_i-\frac{1+nw}{1+(i-1)w}\cdot S_{i-1}\\
		&=\frac{1+nw}{1+iw}\cdot(S_{i-1}+wX_i)-\frac{1+nw}{1+(i-1)w}\cdot S_{i-1}.
	\end{align*}
	\eat{
		\begin{align*}
			&M_i-M_{i-1}\\
			&=\frac{(n-i)a+(1+nw)S_i}{1+iw}-\frac{(n-i+1)a+(1+nw)S_{i-1}}{1+(i-1)w}\\
			&=\frac{(1+nw)\big((1+(i-1)w)X_i-S_{i-1}-a\big)}{(1+iw)(1+(i-1)w)}.
		\end{align*}
	}
	Since $0\leq X_i \leq 1$, we have
	\eat{
		\begin{equation*}
			M_i-M_{i-1}\geq \frac{(1+nw)(-S_{i-1}-a)}{(1+iw)(1+(i-1)w)}\triangleq \Delta_{min}.
		\end{equation*} 
	}
	\begin{equation*}
		M_i-M_{i-1}\leq \frac{1+nw}{1+iw}\cdot(S_{i-1}+w)-\frac{1+nw}{1+(i-1)w}\cdot S_{i-1}\triangleq \Delta_{\max}.
	\end{equation*}
	Similarly, we also have
	\begin{equation*}
		M_i-M_{i-1}\geq \frac{1+nw}{1+iw}\cdot S_{i-1}-\frac{1+nw}{1+(i-1)w}\cdot S_{i-1}\triangleq \Delta_{\min}.
	\end{equation*}
	Hence, 
	\begin{equation*}
		\Delta_{\max} - \Delta_{\min} = \frac{(1+nw)w}{1+iw}.
	\end{equation*}
	Finally, by Azuma inequality \cite{azuma1967weighted}, we have
	\begin{align*}
		&\Pr\big[\abs{M_n-M_0}\geq \gamma\big]\\
		&\leq 2\exp\bigg(-\frac{2\gamma^2}{w^2\sum_{i=1}^n(\frac{1+nw}{1+iw})^2}\bigg)\\
		&\leq 2\exp\bigg(-\frac{2\gamma^2}{w(1+nw)^2\sum_{i=1}^n\big(\frac{1}{1+(i-1)w}-\frac{1}{1+iw}\big)}\bigg)\\
		&=2\exp\bigg(-\frac{2\gamma^2}{w^2(1+nw)n}\bigg).
	\end{align*}
	Setting $\gamma=nwa\varepsilon$ and rearranging it concludes the theorem.
	%
	%
	%
	%
	%
	%
	%
\end{proof}

\begin{proof}[Proof of Theorem~\ref{thm:SL-PoS-robust}]
	Let $X_i\in\{0,1\}$ be a binary random variable indicating whether $A$ is the proposer for the $i$-th block. Let $Z_n = \frac{a+w\sum_{i=1}^{n}X_i}{1+nw}$ be fraction of stakes possessed by $A$ after $n$ blocks. Then, the difference between $Z_{n+1}$ and $Z_n$ can be written as
	\begin{align*}
		Z_{n+1} - Z_{n}
		& = \frac{a+ w\sum_{i=1}^{n+1}X_i}{1+(n+1)w} - Z_{n} \\ 
		& = \frac{(1+nw)Z_n+X_{n+1}w}{1+(n+1)w} -Z_n\\
		& = \frac{w}{1+(n+1)w} \cdot (X_{n+1}-Z_n).
	\end{align*}
	Moreover, let $\gamma_{n+1}=\frac{w}{1+(n+1)w}$, $f(Z_n)= \E[X_{n+1}\mid Z_n]-Z_n$ and $U_{n+1}=X_{n+1}-\E[X_{n+1}\mid Z_n]$. Then,
	\begin{equation*}
		Z_{n+1} - Z_n = \gamma_{n+1}\big(f(Z_n) + U_{n+1}\big).
	\end{equation*}
	
	Next, we verify that conditions \ref{cond1}--\ref{cond4} given in Definition~\ref{def:SA} hold almost surely. For condition~\ref{cond1}, we know that $\frac{w}{(1+w)n}\leq \gamma_n\leq \frac{1}{n}$ and set $c_l={w}/{(1+w)}$ and $c_u=1$. For condition~\ref{cond2}, we set $K_u=1$ as $\abs{U_{n}}\leq 1$. For condition~\ref{cond3}, we know from \eqref{prob:SL-PoS} that 
	\begin{equation}\label{eqn:f(Z_n)}
		f(Z_n)=
		\begin{cases}
			\frac{Z_n}{2(1-Z_n)}-Z_n,&\text{if } Z_n\leq \frac{1}{2},\\
			1-\frac{1-Z_n}{2Z_n}-Z_n,&\text{otherwise}.
		\end{cases}
	\end{equation}
	Thus, it can be seen that $\abs{f(Z_n)}\leq 1$ and hence we set $K_f=1$. Finally, for condition~\ref{cond4}, we find that $\E[\gamma_{n+1} U_{n+1}\mid \mathcal{F}_n]=0$ and hence we set $K_e=0$.
	
	In addition, by \eqref{eqn:f(Z_n)}, we observe that $f(Z_n)$ is continuous for $Z_n\in[0,1]$. Thus, by Lemma~\ref{lemma:sa-zero}, $\lim_{n\to\infty}Z_n$ exists almost surely and is in one of the zeros of $f(\cdot)$. Let $f(x)=0$ such that the zeros are found as $Q_f=\{0,\frac{1}{2},1\}$. Then, it remains to show that $q=1/2$ is an unstable point and $q=0$ and $q=1$ are two stable points. 
	
	Clearly, we have
	\begin{equation*}
		f(x)(x-1/2)=
		\begin{cases}
			\frac{x(x-1/2)}{1-x}\cdot (x-1/2)\geq 0,&\text{if } x\leq \frac{1}{2},\\
			\frac{(1-x)(x-1/2)}{x}\cdot (x-1/2)\geq 0,&\text{otherwise}.
		\end{cases}
	\end{equation*}
	Furthermore,
	\begin{align*}
		\E[U_{n+1}^2\mid \mathcal{F}_n]
		&=\E[X_{n+1}^2\mid Z_n]-\E^2[X_{n+1}\mid Z_n]\\
		&=\E[X_{n+1}\mid Z_n]-\E^2[X_{n+1}\mid Z_n].
	\end{align*}
	Thus, if $Z_n$ is close to $1/2$, \ie~$Z_n\in [1/2-\varepsilon,1/2+\varepsilon]$ for some $\varepsilon>0$, we have $\E[X_{n+1}\mid Z_n]\in [\frac{1/2-\varepsilon}{1+2\varepsilon}, \frac{1/2+3\varepsilon}{1+2\varepsilon}]$. As a result, 
	\begin{equation*}
		\E[U_{n+1}^2\mid \mathcal{F}_n]\geq\frac{1/2-\varepsilon}{1+2\varepsilon}\cdot \frac{1/2+3\varepsilon}{1+2\varepsilon}\triangleq K_L,
	\end{equation*}
	which implies $q=1/2$ is an unstable point. Hence, according to Lemma~\ref{lemma:sa-unstable}, $\Pr[Z_n\to 1/2] = 0$.
	
	Finally, we prove that $q=0$ is a stable point, with $q=1$ being analogous. Obviously, $f(x)x < 0$ when $x>0$ is close to $0$. Meanwhile, if $Z_0=0$, it always holds that $\Pr[Z_n=0]=1$, which implies every neighborhood of $q$ is attainable. Consequently, by Lemma~\ref{lemma:sa-stable}, $\Pr[Z_n \to 0] > 0$. Note that when $Z_n\to 0$, we must have $\lambda_{A}\to 0$. Therefore, when $n\to \infty$, $\Pr[(1-\varepsilon)a\le \lambda_A \le (1+\varepsilon)a] = 0$ for any positive $\varepsilon$, which concludes the theorem.
\end{proof}

\begin{proof}[Proof of Theorem~\ref{thm:C-PoS-robust}]
	The proof is, again, similar to that of Theorem~\ref{thm:ML-PoS-robust} by utilizing Doob's martingale and Azuma's inequality. Let $Y_i$ and $S_i$ be the notations same as those in the proof of Theorem~\ref{thm:C-PoS-fair}. Then, 
	\begin{equation*}
		S_{i+1} = S_{i} + \frac{wY_{i+1}}{P}+ \frac{vS_{i}}{1+(w+v)(i)}.
	\end{equation*}
	Analogous to the analysis for Theorem~\ref{thm:C-PoS-fair}, we have
	\begin{equation*}
		\E[S_{n} \mid S_i] = S_{i} + \frac{(w+v)(n-i)S_i}{1+(w+v)i}= \frac{1+ (w+v)n}{1+(w+v)i}\cdot S_i.
	\end{equation*}
	Moreover, in each epoch, we manually sort the shards and consider the every shard is provided $w/P$ proposer reward and $v/P$ inflation reward. Let $Y_{i,j}\in\{0,1\}$ be a binary random variable indicating whether $A$ is the block proposer for the $j$-th shard of the $i$-th mining epoch. Let $S_{i,j}$ be the number of stakes possessed by $A$ after completing the $j$-the shard of the $i$-the epoch, \eg~$S_{i,P}=S_i$. Conditioned on $S_{i-1}$ and $S_{i,j}$, the expectation of $S_i$ can be computed as
	\begin{equation*}
		\E[S_i\mid S_{i-1},S_{i,j}]=S_{i,j}+\frac{(P-j)(w+v)S_{i-1}}{P(1+(w+v)(i-1))}.
	\end{equation*}
	
	
	Let $M_{i,j} = \E[S_n\mid Y_{1,1},\dotsc,Y_{1,P},Y_{2,1},\dotsc,Y_{2,P},\dotsc,Y_{i,1},\dotsc,Y_{i,j}]$. Obviously, $M_0,M_{1,1},M_{1,2},\dotsc,M_{n,P}$ is a Doob's martingale with respect to $Y_{1,1},Y_{1,2}, \dotsc,Y_{n,P}$. Furthermore, $M_{i,j} $ can be rewritten as 
	\begin{align*}
		M_{i,j} 
		&= \E[S_{n} \mid S_{i-1},S_{i,j}]=\E[\E[S_{n}\mid S_{i}] \mid S_{i-1},S_{i,j}] \\
		&= \frac{1+ (w+v)n}{1+(w+v)i}\cdot\E[S_i\mid S_{i-1},S_{i,j}].
	\end{align*}
	Then, we bound the difference between the maximum and minimum values of martingale difference sequence. In particular, given $S_{i-1}$ and $S_{i,j}$, for any $j\leq P-1$, we have
	\begin{align*}
		\max\{M_{i,j+1}-M_{i,j}\}-\min\{M_{i,j+1}-M_{i,j}\}\leq \frac{1+ (w+v)n}{1+(w+v)i}\cdot \frac{w}{P}.
	\end{align*}
	Meanwhile, we can also get that
	\begin{align*}
		\max\{M_{i,1}-M_{i-1,P}\}-\min\{M_{i,1}-M_{i-1,P}\}\leq \frac{1+ (w+v)n}{1+(w+v)i}\cdot \frac{w}{P}.
	\end{align*}
	Finally, by Azuma inequality, we have 
	\begin{align*}
		&\Pr\big[\abs{M_{n,P}-M_0} \ge \gamma\big]
		\le 2\exp \bigg(-\frac{2\gamma^2}{\sum_{i = 1}^n\sum_{j=1}^{P}\big(\frac{1+ (w+v)n}{1+(w+v)i}\cdot\frac{w}{P} \big)^2}\bigg)\\
		&\le 2\exp \bigg(-\frac{2\gamma^2}{\frac{w^2(1+ (w+v)n)^2}{P(w+v)}\cdot \sum_{i = 1}^n \big(\frac{1}{1+(w+v)(i-1)}-\frac{1}{1+(w+v)i}\big)}\bigg)\\
		&=2\exp \bigg( -\frac{2\gamma^2P}{w^2\big(1+(w+v)n\big)n}\bigg)
	\end{align*}
	Setting $\gamma = na(w+v)\varepsilon$ concludes the theorem.
\end{proof}

\begin{proof}[Proof of Lemma~\ref{lem:C-PoS-multi-player}]Suppose that there are $m$ miners. Denote by $S^i$ the fraction of stakes possessed by miner $i$ such that $\sum_{k=1}^{m}S^i=1$ and by $T^i$ the waiting time of miner $i$'s candidate block becoming valid. Without loss of generality, we assume that $S^1\leq S^2\leq\dotsb\leq S^m$. As discussed in Section~\ref{subsec:SL-PoS}, $T^i=\mathtt{basetime} \cdot X^i/S^i$, where $X^i$ is a random hash value uniformly distributed in the range of $[0,2^{256}-1]$ such that $\frac{X^i}{2^{256}}$ follows the continuous uniform distribution $U(0,1)$ asymptotically. Let $Z^i=\frac{X^i}{2^{256}\cdot S^i}$ such that $Z^i\sim U(0,\frac{1}{S^i})$. Then, given $Z^i=z$, we have
\begin{equation*}
\Pr\Big[\bigwedge\nolimits_{j\neq i} (Z^j\geq z)\Big]=\prod\nolimits_{j\neq i}(1-S^jz)^{+},
\end{equation*}
where $(1-S_jz)^{+}=\max\{1-S_jz,0\}$. Therefore, the probability of miner $i$ winning the next block is
\begin{align*}
	&\Pr\Big[\bigwedge\nolimits_{j\neq i} (T^j\geq T^i)\Big]
	=\Pr\Big[\bigwedge\nolimits_{j\neq i} (Z^j\geq Z^i)\Big]\\
	&=\int_{0}^{\frac{1}{S^i}}S^i\prod\nolimits_{j\neq i}(1-S^jz)^{+}\diff z
	=\int_{0}^{\frac{1}{S^m}}S^i\prod\nolimits_{j\neq i}(1-S^jz)\diff z.
\end{align*}
We consider miner $1$ with the minimum staking power. We have
\begin{align*}
	&\Pr\Big[\bigwedge\nolimits_{j\neq i} (T^j\geq T^1)\Big]
	=\int_{0}^{\frac{1}{S^m}}S^1\prod\nolimits_{j=2}^{m}(1-S^jz)\diff z\\
	&\leq \int_{0}^{\frac{m-1}{1-S^1}}S^1\Big(1-\frac{1-S^1}{m-1}\cdot z\Big)^{m-1}\diff z
	=\frac{m-1}{m}\cdot \frac{S^1}{1-S^1}
	\leq S^1,
\end{align*}
where the first inequality is because the maximum is achieved at $S^2=\dotsb=S^m=(1-S^1)/m$ and the second inequality is from the fact that $\frac{1}{1-S^1}\leq \frac{m}{m-1}$ since $S^1\leq 1/m$. Moreover, in the above inequality, ``$=$'' holds if and only if $S^1=S^2=\dotsb=S^m=1/m$, and when $S^1<1/m$, such a probability is less than $S^1$. This completes the proof.
\end{proof}

\end{sloppy}
\end{document}